\documentclass[12pt]{article}

\RequirePackage[OT1]{fontenc}
\RequirePackage{amsthm,amsmath,amsmath,amsfonts,amssymb,mathabx}
\RequirePackage{bm}
\RequirePackage{url}
\RequirePackage{natbib}
\usepackage[usenames]{color}
\RequirePackage[colorlinks,citecolor=blue,urlcolor=blue]{hyperref}
\usepackage{times}
\usepackage{appendix}
\usepackage{multirow}
\usepackage{natbib}
\usepackage{enumitem}
\usepackage{footmisc}
\usepackage{graphicx}
\usepackage{subfigure}
\usepackage{makecell}
\usepackage{booktabs}
\usepackage{booktabs}
\usepackage{array}
\usepackage{url}

\usepackage[ruled]{algorithm2e}
\usepackage{algorithmic}
\usepackage{dsfont}
\usepackage{authblk}
\usepackage{hyperref}

\allowdisplaybreaks[4]

\newcommand{\bfsym}[1]{\ensuremath{\boldsymbol{#1}}}

\def\balpha{\bfsym \alpha}
\def\bbeta{\bfsym \beta}

\def\bxi{\bfsym \xi}

\def\bgamma{\bfsym \gamma}

\def\bSigma{\bfsym \Sigma}
     \def\bA{\bfsym A}
\def\bb{\bfsym b}     \def\bB{\bfsym B}
\def\bc{\bfsym c}     
     
\def\be{\bfsym e}     
\def\bff{\bfsym f}    \def\bF{\bfsym F}
     
     \def\bH{\bfsym H}
     \def\bI{\bfsym I}

     \def\bM{\bfsym M}

     \def\bP{\bfsym P}

     \def\bS{\bfsym S}
     
\def\bu{\bfsym u}     \def\bU{\bfsym U}
\def\bv{\bfsym v}     \def\bV{\bfsym V}
\def\bw{\bfsym w}     
\def\bx{\bfsym x}     \def\bX{\bfsym X}
     \def\bY{\bfsym Y}
     \def\bZ{\bfsym Z}

\def\cA{{\cal  A}}

\def\cC{{\cal  C}}
\def\cD{{\cal  D}}
\def\cE{{\cal  E}}

\def\cI{{\cal  I}}

\def\cK{{\cal  K}}

\def\cS{{\cal  S}}

\def\cV{{\cal  V}}

\def \bbP {\mathbb{P}}

\def \bbE {\mathbb{E}}
\def \bbF {\mathbb{F}}
\def \bbR {\mathbb{R}}

\renewcommand{\hat}{\widehat}
\renewcommand{\tilde}{\widetilde}

\def\hbbeta{\hat{\bbeta}}

\def\hbgamma{\hat{\bgamma}}

\def \bvarphi      {\bfsym {\varphi}}

\def \bPhi      {\bfsym {\Phi}}

\def \bomega    {\bfsym {\omega}}
\def \bphi      {\bfsym {\phi}}

\def \bnu      {\bfsym {\nu}}

\def\bSigma{\bfsym \Sigma}
\def \bdelta    {\bfsym {\delta}}
\def \bDelta    {\bfsym {\Delta}}
\def \bTheta    {\bfsym {\Theta}}
\def \bGamma    {\bfsym {\Gamma}}
\def\blambda    {\bfsym {\gamma}}
\def\bxi        {\bfsym {\xi}}

\DeclareMathOperator{\Var}{Var}

\DeclareMathOperator{\Cov}{Cov}

\def \etal {{\em et al.}}

\newcommand{\tr}[0]{\mathrm{tr}}

\theoremstyle{plain}% default
\newtheorem{theorem}{Theorem}[section]
\newtheorem{lemma}[theorem]{Lemma}
\newtheorem{proposition}[theorem]{Proposition}

\theoremstyle{definition}

\newtheorem{example}{Example}[section]
\newtheorem{assumption}{Assumption}[section]

\newtheorem{remark}{Remark}

\numberwithin{equation}{section}

\definecolor{DSgray}{cmyk}{0,1,0,0}

\makeatletter
\def\singlespace{\def\baselinestretch{1}\@normalsize}

\newcommand{\blind}{1}

\begin{document}
	
	\def\spacingset#1{\renewcommand{\baselinestretch}%
		{#1}\small\normalsize} \spacingset{1}

	\if1\blind
	{
		\title{ Are  Latent Factor Regression and Sparse Regression Adequate?}
		\author{  Jianqing Fan \qquad Zhipeng Lou \qquad  Mengxin Yu  }
		\date{}
		\maketitle
		\begin{singlespace}
			\begin{footnotetext}[1]
				{%{\bf Add all authors affiliations here}.
					Jianqing Fan is Frederick L. Moore '18 Professor of Finance, Professor of Statistics, and Professor of Operations Research and Financial Engineering at the Princeton University.  Zhipeng Lou is a   Postdoctoral Researcher at Department of Operations Research and Financial Engineering, Princeton University. Mengxin Yu is a Ph.D. student at Department of Operations Research and Financial Engineering, Princeton University, Princeton, NJ 08544, USA.   
					Emails: \texttt{\{jqfan, zlou,  mengxiny\}@princeton.edu}.  The research is supported  in part by the NIH grant 2R01-GM072611-16,  and the NSF grants DMS-1712591, DMS-2052926, DMS-2053832, and the ONR grant N00014-19-1-2120				%The authors are indebted to the referees, the associate editor and the Co-editor for their
					%valuable comments, which have significantly improved the paper.
				}
			\end{footnotetext}
			
		\end{singlespace}
	}
	\fi
	\if0\blind
	{
		\title{ Are  Latent Factor Regression and Sparse Regression Adequate?}
		\author{}
		\date{}
		\maketitle
	} \fi

\begin{abstract}
We propose the \textbf{F}actor \textbf{A}ugmented sparse linear \textbf{R}egression \textbf{M}odel (FARM) that not only encompasses both the latent factor regression and sparse linear regression as special cases but also bridges dimension reduction and sparse regression together. We provide theoretical guarantees for the estimation of our model under the existence of sub-Gaussian and heavy-tailed noises (with bounded $(1+\vartheta)$-th moment, for all $\vartheta>0$) respectively. In addition, the existing works on supervised learning often assume the latent factor regression or sparse linear regression is the true underlying model without justifying its adequacy. To fill in such an important gap, we also leverage our model as the alternative model to test the sufficiency of the latent factor regression and the sparse linear regression models. To accomplish these goals, we propose the \textbf{F}actor-\textbf{A}djusted de\textbf{B}iased \textbf{Test} (FabTest) and a two-stage ANOVA type test respectively. We also conduct large-scale numerical experiments including both synthetic and FRED macroeconomics data to corroborate the theoretical properties of our methods. Numerical results illustrate the robustness and effectiveness of our model against latent factor regression and sparse linear regression models.

\end{abstract}

\newpage
\spacingset{1.45} % DON'T change the spacing!

\pagestyle{plain}

\iffalse
\begin{abstract}
We propose the \textbf{F}actor \textbf{A}ugmented   sparse linear \textbf{R}egression \textbf{M}odel (FARM) that not only encombances both the latent factor regression and sparse linear regression as special cases but also bridges  dimension reduction and sparse regression together. We provide theoretical guarantees for the estimation of our model under the existence of sub-Gaussian and heavy-tailed noises (with bounded $(1+\vartheta)$-th moment, for all $\vartheta>0$) respectively. In addition, the existing works on supervised learning often assume the latent factor regression or the sparse linear regression is the true underlying model without justifying its adequacy. To fill in such an important gap, we also leverage our model as the alternative model to test the sufficiency of the latent factor regression and the sparse linear regression models. To accomplish these goals, we propose the \textbf{F}actor-\textbf{A}djusted de\textbf{B}iased \textbf{Test} (FabTest) and a two-stage ANOVA type test respectively. We also conduct large scale numerical experiments including both synthetic and FRED macroeconomics data to corroborate the theoretical properties of our methods. Numerical results illustrate the robustness and effectiveness of our model against latent factor regression and sparse linear regression models.
\end{abstract}
\fi
\section{Introduction}
%\textcolor{red}{Introduction to cross-sectional correlation issue.}

%In this paper, we consider the fundamental regression problem for $(\bx, Y)$, where $\bx$ is a $d$-dimensional vector of covariates and $Y \in \bbR$ is the response. Our goal is to propose a linear model which can incorporate the dependence structure.

Over the past two decades, along with the development of technology, datasets with high-dimensionality in various fields such as biology, genomics, neuroscience and finance have been collected. One stylized feature of the high-dimensional data is the high co-linearity across features. A common structure to characterize the dependence across features is the approximate factor model~\citep{Bai2003inferential, Fan2013POET}, in which the variables are correlated with each other through several common latent factors. More specifically, we assume the observed $d$-dimensional covariate vector $\bx$ follows from the model
\begin{align}
\label{eq_factor_model}
    \bx = \bB \bff + \bu,
\end{align}
where $\bff$ is a $K$-dimensional vector of latent factors, $\bB\in \mathbb{R}^{d\times K}$ is the corresponding factor loading matrix, and $\bu$ is a $d$-dimensional vector of idiosyncratic component which is  uncorrelated with $\bff$.

%the We focus on the case in which the dependence structure is characterized by a factor model in which the variables are correlated with each other through several latent factors. More specifically, we assume the observed $d$-dimensional covariate vector $\bx$ follows from the model
%\begin{align*}
%    \bx = \bB \bff + \bu,
%\end{align*}
%where $\bff$

To tackle the high-dimensionality of datasets, various methods have been proposed. Among these, dimensionality reduction and sparse regression are two popularly used ones to circumvent the curse of dimensionality.  They also serve as the backbones for many emerging statistical methods.

%the factor regression model, in which the latent factor $\bff$ contributes to both the predictor $\bx$ and the response $Y$.

In terms of dimension reduction, the factor regression model is one of the most popular methods and has been widely used~\citep{Stock_Watson, Bai2006, Bair_Hastie_2006, Bai2008, FAN_Xue_Yao, Benea2020_2, Benea2020, Bing2021}.  It assumes that the factors drive both dependent and independent variables as follows:
\begin{align}
\label{eq3}
    Y &= \bff^{\top}\bgamma + \varepsilon,\cr
    \bx &= \bB \bff + \bu.
\end{align}
Here $Y$ is the response variable and $\varepsilon \in \bbR$ is the random noise which is independent with the factor $\bff$. When the factors are unobserved, one usually learns the latent factors based on observed $\bx$ and substitutes the sample version into the regression model~\eqref{eq3}. There are several methods for estimating latent factors such as  Principal Component Analysis (PCA)~\citep{Bai2003inferential, Fan2013POET}, maximum likelihood estimation~\citep{Bai2012statistical}, and random projections~\citep{Liao2020}. In particular, when the leading Principal Components are used as an estimator for $\bff$, the sample version of \eqref{eq3} reduces to the classical Principal Component Regression (PCR) \citep{Hotelling1933analysis}.
%the By picking out the leading principal components, PCR serves as an efficient dimensional reduction method and preserves enough explanation for the variance of the response $Y$.

%For more details, please see...
%For a real dataset, the PCR model \eqref{eq3} holds when the response $Y$ is fully explained by the leading principal components of $\bx$. Under many real data applications, the PCR model serves as an efficient dimension reduction method is efficient in supervised learning with strongly correlated features (e.g., features that share some common factors), as in many cases, the leading principal components have already gained sufficient explanation power. \textcolor{red}{cite some real data examples.}
%\textcolor{red}{cite some other related literatures and recent progress related to PCR(later).}
% However, the RCR model relies on the assumption only the leading principal components contribute to the response, which

As for sparse regression, a commonly used model is the following (sparse) linear regression:
\begin{align}
\label{eq4}
    Y = \bx^{\top} \bbeta + \varepsilon.
\end{align}
In the high dimensional regime where the dimension $d$ can be much larger than the sample size $n$, it is commonly assumed that the population parameter vector $\bbeta \in \bbR^{d}$ is sparse. Over the last two decades, various regularized methods, which incorporate this notion of sparsity, have been proposed. See, for instance, LASSO~\citep{LASSO1996}, SCAD~\citep{SCAD2001}, Least Angle Regression~\citep{Efron2004}, Dantzig selector~\citep{CandesTao2007}, Adaptive LASSO~\citep{ALASSO2006}, MCP~\citep{MCP2010} and many others. For more details, please refer to~\cite{Fan_Li_Zhang_Zou_data} for a comprehensive survey.

\iffalse
Here the population parameter vector $\bbeta\in\bbR^{d}$ is assumed to be an unknown $s$-sparse parameter with $s=o(n)$ that one aims at recovering. In the high dimensional regime, various regularized estimation of $\bbeta$ have been proposed in the past two decades \citep{LASSO1996, SCAD2001, Efron2004, CandesTao2007, ALASSO2006, MCP2010}. %Imposing such sparsity assumption helps %For instance, see \citet{LASSO1996,SCAD2001,Efron2004,CandesTao2007,ALASSO2006,MCP2010}.%, Adaptive LASSO \citet{ALASSO2006}, MCP \citet{MCP2010} and many others.
\fi

In this paper, we introduce the \textbf{F}actor \textbf{A}ugmented sparse linear \textbf{R}egression \textbf{M}odel (FARM)~\eqref{factor_aug}, which incorporates both the latent factor and the idiosyncratic component into the covariates,
%which contains PCR and sparse linear regression as our special case.
\iffalse
\textcolor{blue}{
\begin{align}
\label{factor_aug}
    Y &= \bff^{\top}\bvarphi^{\star} + \bx^{\top} \bbeta^{\star} + \varepsilon, \enspace \bx=\bB\bff+\bu, \textrm{\,\,which is equivalent with}%\mathrm{where} \enspace \bvarphi^{\star} = \bgamma^{\star} - \bB^\top \bbeta^{\star} \in \bbR^{K}.
    \\
    Y & %\bff^{\top}(\bvarphi^{\star}+\bB^\top \bbeta^{*}) + \bu^{\top} \bbeta^{\star} + \varepsilon
    =\bff^{\top}\bgamma^{\star} + \bu^{\top} \bbeta^{\star} + \varepsilon \label{factor_aug2}
\end{align}}
\fi
\begin{align}
\label{factor_aug}
    %\bx = \bB\bff + \bu \enspace \mathrm{and} \enspace
    Y &= \bff^{\top}\bgamma^{\star} + \bu^{\top}\bbeta^{\star} + \varepsilon,\cr
    \bx &= \bB \bff + \bu,
\end{align}
where $\bgamma^{\star} \in \bbR^{K}$ and $\bbeta^{\star} \in \bbR^{d}$ are population parameter vectors quantifying the contribution of the latent factor $\bff$ and the idiosyncratic component $\bu$, respectively. Obviously, the factor regression model~\eqref{eq3} is a special case of~\eqref{eq_factor_linear} in which $\bbeta^{\star} = 0$. To better illustrate the difference between model~\eqref{factor_aug} and the sparse linear model~\eqref{eq4}, our model can be written in an equivalent form,
\begin{align}
\label{eq_model_varphi}
    Y &= \bff^{\top}\bvarphi^{\star} + \bx^{\top} \bbeta^{\star} + \varepsilon,\cr
    \bx &= \bB \bff + \bu,
\end{align}
where $\bvarphi^{\star} = \bgamma^{\star} - \bB^{\top}\bbeta^{\star} \in \bbR^{K}$ quantifies the extra contribution of the latent factor $\bff$ beyond the observed predictor $\bx$. Therefore, FARM expands the space spanned by $\bx$ into useful directions spanned by $\bff$. It is clear that the sparse regression model~\eqref{eq4} is also a special case of~\eqref{factor_aug} with $\bvarphi^{\star} = 0$.  Thus, our model is general enough to bridge the dimensionality reduction and the sparse regression.

%\textcolor{blue}{
%Here $\bff$ servers as the factor component which provides explanation for both the predictor $\bx$ and the response $Y$.
%}
%We replace $\bvarphi^{\star}+\bB^\top \bbeta^{*}$ with $\bgamma^{*}\in \mathbb{R}^{K}$ which is independent with $\bbeta^{*}$ in \eqref{factor_aug2} since there exists a  bijection between $(\bgamma^{*},\bbeta^{*})$ and $(\bvarphi^*,\bbeta^{*})$.

The motivation of our factor augmented linear model~\eqref{factor_aug} comes from two perspectives.
\begin{enumerate}
%\item \textcolor{red}{Introduce factor model. Why factors are important.}
    \item Firstly, it origins from \cite{Fan2020factor}. In order to get precise estimation of $\bbeta^{\star}$ based on highly correlated variables, they study the sparse regression estimation by substituting~\eqref{eq_factor_model} into~\eqref{eq4} and obtain
    \begin{align}
    \label{eq_linear_model}
        Y = (\bB\bff + \bu)^{\top}\bbeta^{\star} + \varepsilon = \bff^{\top}(\bB^{\top}\bbeta^{\star}) + \bu^{\top}\bbeta^{\star} + \varepsilon.
    \end{align}
We observe from~\eqref{eq_linear_model}, when the sparse linear regression is adequate, for a given $\bbeta^{\star},$ the regression coefficient on $\bff$ is fixed at $\bgamma^{\star} = \bB^{\top}\bbeta^{\star}$. However, in reality, especially when the variables are highly correlated, it is very likely that the leading factors possess extra contributions to the response instead of only a fixed portion $\bB^{\top}\bbeta^{\star}$. This results in our proposition of model~\eqref{factor_aug}, where we augment the leading factors into sparse regression that expands the linear space spanned by $\bx$ into useful directions.

\item Secondly, it origins from the factor regression given in~\eqref{eq3}. %As we discussed above, factor regression together with its variants is favored as an effective dimension reduction method when one has observed highly correlated features with high-dimensionality \citep{Stock_Watson,Bai2006,Bair_Hastie_2006,Bai2008,FAN_Xue_Yao,Benea2020_2, Benea2020,Bing2021}.
In reality, the leading common factors $\bff$ indeed provides some important contributions to the response, but it is hard to believe that they will have fully explanation power, especially when the effect of the factors is weak. Besides, in real applications, several examples illustrate the poor performance of factor regression model or PCR, see~\cite{Jolliffe1982} for more details. Thus, completely ignoring the idiosyncratic component $\bu$ will harm in model generalization. This also motivates us to propose model~\eqref{factor_aug}, in which we augment the sparse regression by incorporating the idiosyncratic component $\bu$ into the original factor regression.

%\item \color{red}{From the third perspective, it origins from the field of causal inference. (Hidden confounding often exists, give some examples) When there exists some hidden confounding that affects both covariate $\bx$ and response $Y$, using classical linear regression will results in a bad estimator with large bias. (cite some related paper) This motivates us to depict the hidden confounding as the hidden factor $\bff$ of $\bx$ that contributes to both $\bx$ and $Y$. Thus, after incorporating the hidden confounding variable, we obtain the model \eqref{eq_model_varphi}}. %be the hidden confounding that contributes to both $\bx$ and $Y$.

\end{enumerate}

%Going beyond the linear regression model and the least squares estimation, our idea of incorporating the dependence structure within the covariates to the linear regression model can be naturally extended to more general supervised learning models through different loss functons. For instance, quantile regression \citep{Belloni2011,Fan2014adaptive}, support vector machine \citep{Zhang2016variable,Peng2016}, Huber regression \citet{Fan2017estimation,AHR2020}, generalized linear model \citet{Van2008high,Fan2020factor} and many other variants.
In this paper, we first study the properties of estimated parameters under the proposed model~\eqref{factor_aug}. Specifically, we assume the factors given in~\eqref{factor_aug} are unobserved and leverage PCA to estimate them. Incorporated with penalized least-squares with the $\ell_{1}$-penaly, we derive the $\ell_{2}$-consistency results for parameter vectors $\bgamma^{\star}$ and $\bbeta^{\star}$. Going beyond the linear regression model and the least squares estimation, our idea can be naturally extended to more general supervised learning models through different loss functons. For instance, quantile regression~\citep{Belloni2011, Fan2014adaptive}, support vector machine~\citep{Zhang2016variable, Peng2016}, Huber regression~\citep{Fan2017estimation, AHR2020}, generalized linear model~\citep{Van2008high, Fan2020factor} and many other variants. In order to demonstrate the general applicability of our proposed methods, in our paper, we further extend our model settings to robust regression. To be more specific, we only assume the existence of $(1 + \vartheta)$-th moment of the noise distribution for some $\vartheta > 0$. We adopt Huber loss together with adaptive tuning parameters and $\ell_{1}$-penalization to derive the consistency results for the parameters of our interest. Besides the aforementioned extensions, it is worth to note that our model is also applicable in the field of causal inference~\citep{Imbens2015,hernan2019causal}. To be more specific, the latent factors $\bff$ given in our model are able to be treated as the unobserved confounding variables which affect both the covariate $\bx$ and the response $Y$. From the causal perspective, we provide a methodology to conduct (robust) statistical estimation as well as inference of our model under the existence of latent confounding variables.

The aforementioned works on factor regression and sparse linear regression mainly investigate  the theoretical properties based on the assumption that either of them is the true underlying model~\citep{Stock_Watson, LASSO1996, SCAD2001, ALASSO2006, Bai2006, MCP2010, FAN_Xue_Yao, Fan2020factor, Bing2021}. %The work on studying their adequacy, i.e., whether factor regression and sparse regression are sufficient, especially under highly correlated measurements.
However, whether a given model is adequate to explain a given dataset plays a crucial role in the model selection step. This  motivates us to fill the gap by leveraging our model as the alternative one to perform hypothesis testing on the adequacy of the factor regression model as well as the sparse linear regression model when covariates admit a  factor structure.

%As we mention above, our model \eqref{eq_model_varphi} is general enough that contains factor regression and sparse linear regression as the special cases.
%It turns out when $\bx$ encompasses the factor structure i.e. $\bx=\bB\bff+\bu$, model \eqref{factor_aug} and $\eqref{factor_aug2}$ are equivalent.
%To be more specific, we obtain the $\ell_2$-convergence rate of our estimators $\hat\bbeta_{\lambda}$ and $\hat\bgamma$ to $\bbeta^{*}$ and $\bgamma^{*}$.
%As we mentioned above, literature that aims at quantifying the adequacy of factor regression and sparse regression is sparse. To fill in the blank page, we leverage model~\eqref{factor_aug} as the alternative to conduct hypothesis on the adequacy of the factor regression and sparse linear regression.

For the hypothesis test on the adequacy of the latent factor regression model, we consider testing the hypotheses
\begin{align}
\label{test_pcr}
    H_{0} : Y = \bff^{\top}\bgamma^{\star} + \varepsilon \enspace \mathrm{versus} \enspace H_{1} : Y = \bff^{\top}\bgamma^{\star} + \bu^{\top} \bbeta^{\star} + \varepsilon.
\end{align}
This amounts to testing $H_0:  \bbeta^{\star}= 0$ under FARM model.
To this end, we propose the \textbf{F}actor-\textbf{A}djusted de\textbf{B}iased \textbf{Test} statistic (FabTest) $\tilde{\bbeta}_{\lambda}$ which serves as a de-sparsify version of the estimator $\hat\bbeta_{\lambda}$ obtained under $\ell_{1}$-regularization. The asymptotic distribution of the proposed test statistic is derived by leveraging the high-dimensional Gaussian approximation. The critical value controlling the Type-I error is estimated based on the  multiplier bootstrap method. As a byproduct, we are also able to conduct entrywise and groupwise hypothesis testing on parameter $\bbeta^{\star}$ by following similar de-biasing procedure.

% It is worth to note that following similar de-biasing procedure, we are also able to conduct entry-wise, group-wise hypothesis testing of our model~\eqref{factor_aug} (as did in~\citet{Javanmard2014} and~\citet{Zhang2017}).

For validating the adequacy of the sparse linear regression model, we consider testing the hypotheses
\begin{align}
\label{Test_sparse_model}
    H_0 : Y = \bx^{\top} \bbeta^{\star} + \varepsilon \enspace \mathrm{versus} \enspace H_1 : Y = \bff^{\top} \bvarphi^{\star} + \bx^{\top} \bbeta^{\star} + \varepsilon,
\end{align}
or $\bvarphi^\star = 0$ under the FARM model.
To tackle the testing problem, we propose a two-stage ANOVA test. In the first stage, we use marginal screening~\citep{Fan2008sure} to pre-select a group of variables which cope well the curse of high dimensionality. In the second stage, we derive the ANOVA-type test statistic. Asymptotic null distribution and the power of the test statistic are derived. In addition, we further extend the aforementioned two-stage ANOVA test to linear multi-modal models~\citep{Li2021}, whose data framework has been well applied in a wide range of scientific fields (e.g multi-omics data in genomics, multimodal neuroimaging data in neuroscience, multimodal electronic health records data in health care).

In summary, our main contributions are as follows:
\begin{enumerate}
    \item Motivated from the factor regression and sparse regression, we propose the \textbf{F}actor \textbf{A}ugmented (sparse linear) \textbf{R}egression \textbf{M}odel (FARM)~\eqref{factor_aug} [also~\eqref{eq_model_varphi}] and investigate in the parameter estimation properties on $\bgamma^{\star}$ and $\bbeta^{\star}$ given in \eqref{factor_aug}. Our work serves as an extension of~\cite{Fan2020factor} to a general setting with weaker assumptions.   It augments the sparse linear regression in useful directions of common factors.
    %In addition, we don't impose the assumption that the sparse linear regression is the true underlying model as did in \cite{Fan2020factor}.
    \item To further demonstrate the wide applicability of our methods, we extend our model to a more robust setting, where we only assume the existence of $(1+\vartheta)$-th moment ($\vartheta>0$) of our noise distribution. Leveraging the  $\ell_1$-penalized adaptive Huber estimation, we establish statistical estimation results for our parameters of interest. Comparing with those closely related literature \citep{Fan2020factor,Fan2021_factor}, our assumption on the moment condition of the noise variable is the weakest. Our robustified factor augmented regression also serves as an extension of \cite{AHR2020} to a more general setting.
    \item In terms of testing the adequacy of the factor regression, we propose  the \textbf{FabTest} by incorporating the factor structure into the de-biased estimators~\citep{Sara2014, Zhang2014, Javanmard2014}. Accompanied with Gaussian approximation, the asymptotic distribution of our test statistic is derived. As for implementation, we propose the multiplier bootstrap method to estimate the critical value in order to control the Type-I error.
    \item For testing the adequacy of sparse linear regression model, we propose a two stage ANOVA-type testing procedure. Asymptotic distribution (under the null) and power (under the alternative) of our constructed test statistic are investigated. In addition, we further extend the methodology to the multi-modal sparse linear regression model~\citep{Li2021}, by testing whether the sparse linear regression for some given modals is adequate.
    \item We conduct large scale simulation studies for our proposed methodology using both synthetic data and real data. Simulation results via synthetic data lend further support to our theoretical findings. As for real data, we  apply our methodology to the studies of the macroeconomics dataset named FRED-MD~\citep{McCracken2016}. The experimental results also illustrate the high efficiency and robustness of our model (FARM) against latent factor regression as well as sparse linear regression.
\end{enumerate}

%Introduction outline:
%\begin{enumerate}
%    \item introduction to principal component regression.
%    \item introduction to sparse regression.
%    \item covariate with factor structures.
%    \item introduction to the new model, intuition and benefit.
%\end{enumerate}
%Our work contributes to the line of principal component regression, sparse regression as well as the growing literature on factor adjust estimation and inference. We’ll briefly review the related works in the below
%\subsection{Related Literature}
%\begin{enumerate}
 %   \item related to principal component regression
%    \item related to sparse regression
 %   \item several closely related literature.
%\end{enumerate}

\subsection{Notation}
For a vector $\bgamma = (\gamma_{1}, \ldots, \gamma_{m})^{\top} \in \bbR^{m}$, we denote its $\ell_{q}$ norm as $\|\bgamma\|_{q} = (\sum_{\ell = 1}^{m} |\gamma_{\ell}|^{q})^{1/q}$, $1\leq q < \infty$, and write $\|\bgamma\|_{\infty} = \max_{1\leq \ell \leq m} |\gamma_{\ell}|$. For any integer $m$, we denote $[m] = \{1, \ldots, m\}$. The sub-Gaussian norm of a scalar random variable $Z$ is defined as $\|Z\|_{\psi_{2}} = \inf \{t > 0 : \bbE \exp(Z^{2}/t^{2}) \leq 2\}$. For a random vector $\bx \in \bbR^{m}$, we use $\|\bx\|_{\psi_{2}} = \sup_{\|\bv\|_{2} = 1} \|\bv^{\top} \bx\|_{\psi_{2}}$ to denote its sub-Gaussian norm. Let $\mathbb{I}\{\cdot\}$ denote the indicator function and let $\mathbf{I}_K$ denotes the identity matrix in $\mathbb{R}^{K\times K}$.  For a matrix $\bA = [A_{jk}]$, we define $\|\bA\|_{\bbF} = \sqrt{\sum_{jk} A_{jk}^{2}}$, $\|\bA\|_{\max} = \max_{jk} |A_{jk}|$ and $\|\bA\|_{\infty} = \max_{j}\sum_{k} |A_{jk}|$ to be its Frobenius norm, element-wise max-norm and matrix $\ell_{\infty}$-norm, respectively. Moreover, we use $\lambda_{\min}(\bA)$ and $\lambda_{\max}(\bA)$ to denote the minimal and maximal eigenvalues of $\bA$, respectively. We use $|\cA|$ to denote the cardinality of set $\cA$. For two positive sequences $\{a_n\}_{n\ge 1}$, $\{b_n\}_{n\ge 1}$, we write $a_n=O(b_n)$ if there exists a positive constant $C$ such that $a_n\le C\cdot b_n$ and we write $a_n=o(b_n)$ if $a_n/b_n\rightarrow 0$. In addition, $a_n=O_{\mathbb{P}}(b_n)$ and $a_n=o_{\mathbb{P}}(b_n)$ have similar meanings as above except that the relationship of $a_n/b_n$ holds with high probability.

\subsection{RoadMap}
The rest of this paper is organized as follows. We study the parameter estimation properties of our proposed model (FARM) in section~\ref{est_factor}, where theoretical results of both regular and robust estimators are analyzed. In section~\ref{pcr_test}, we construct a de-biased test statistic to test the adequacy of latent factor regression model. In addition, in section~\ref{sparse_test}, we construct a two-stage ANOVA test to study the adequacy of sparse linear regression under the setting with highly correlated features. Moreover, to corroborate our theoretical findings, in section~\ref{simulation}, we conduct exhaustive simulation studies. Last but not least, we apply our methodology to study the real data FRED-MD in section~\ref{real_data}.

\section{Factor Augmented Regression Model}\label{est_factor}
\iffalse Factor augmented sparse linear regression model~\eqref{factor_aug} encompasses both the PCR and sparse regression models as its specific case.  To proceed our analysis on doing hypothesis testing on PCR and sparse regression, we first study the properties of our proposed alternative model \eqref{factor_aug}. Thus, in this section, we study the theoretical guarantees of model \eqref{factor_aug} under some standard assumptions.
%The only difference is that we use $\hat{\bu}_t$ and $\hat{\bff}_t$, but this has been handled by Fan \etal (2020c).  So, I think that the results are straightforward.
\fi

%In this section, we introduce the factor augmented sparse linear model and investigate the corresponding theoretical properties on parameter estimation.

The primary objective of this section is to propose a regularized estimation method for our factor augmented sparse linear model and investigate the corresponding statistical properties. Suppose we observe $n$ independent and identically distributed (i.i.d.)~random samples $\{(\bx_{t}, Y_{t})\}_{t = 1}^{n}$ from $(\bx, Y)$, which satisfy that
\begin{align}
\label{eq_factor_linear}
    \bx_{t} = \bB \bff_{t} + \bu_{t} \enspace \mathrm{and} \enspace Y_{t} = \bff_{t}^{\top}\bgamma^{\star} + \bu_{t}^{\top} \bbeta^{\star} + \varepsilon_{t}, \quad t = 1, \ldots, n,
\end{align}
where $\bff_{1}, \ldots, \bff_{n}\in\bbR^{K}$, $\bu_{1}, \ldots, \bu_{n} \in \bbR^{d}$ and $\varepsilon_{1}, \ldots, \varepsilon_{n} \in \bbR$ are i.i.d.~realizations of $\bff$, $\bu$ and $\varepsilon$, respectively. To ease the presentation, we rewrite~\eqref{eq_factor_linear} in a more compact matrix form as follows,
\begin{align}
\label{eq_model_matrix}
    \bX &= \bF \bB^{\top} + \bU,\cr
    \bY &= \bF \bgamma^{\star} + \bU \bbeta^{\star} + \cE,
\end{align}
where $\bX = (\bx_{1}, \ldots, \bx_{n})^{\top}$, $\bF = (\bff_{1}, \ldots, \bff_{n})^{\top}$, $\bU = (\bu_{1}, \ldots, \bu_{n})^{\top}$, $\bY = (Y_{1}, \ldots, Y_{n})^{\top}$ and $\cE = (\varepsilon_{1}, \ldots, \varepsilon_{n})^{\top}$. Throughout the whole paper, we assume we only get access to observations $\{(\bx_{t}, Y_{t})\}_{t = 1}^{n}$. Both the latent factors $\bF$ and the idiosyncratic components $\bU$ are unobserved and need to be estimated from the observed predictors $\bX$. Thus, in the following, we shall first illustrate how to estimate $\bF$ and $\bU$ and then proceed with the regularized estimation for model~\eqref{eq_model_matrix}.

\subsection{Factor Estimation}
\label{fac_est}
\iffalse
\textcolor{blue}{Suppose we observe $n$ independent and identically distributed (i.i.d.)~random samples $\bx_{1}, \ldots, \bx_{n} \in \bbR^{d}$ from the factor model~\eqref{eq_factor_model} in which
\begin{align*}
    \bx_{t} = \bB \bff_{t} + \bu_{t}, \enspace t \in [n].
\end{align*}
Here $\bff_{1}, \ldots, \bff_{n}$ and $\bu_{1}, \ldots, \bu_{n}$ are i.i.d.~realizations of $\bff$ and $\bu$, respectively.
%Here $\{\bff_{t}\}_{t = 1}^{n}$ are i.i.d.~$K$-dimensional random vectors of latent factors and $\{\bu_{t}\}_{t = 1}^{n}$ are i.i.d.~$d$-dimensional vectors of idiosyncratic component.
To ease the presentation, we rewrite the factor model in a compact matrix form as follows,
\begin{align*}
    \bX = \bF \bB^{\top} + \bU,
\end{align*}
where $\bX = (\bx_{1}, \ldots, \bx_{n})^\top \in \bbR^{n \times d}$, $\bF = (\bff_{1}, \ldots, \bff_{n})^\top \in \bbR^{n \times K}$ and $\bU = (\bu_{1}, \ldots, \bu_{n})^\top \in \bbR^{n \times d}$. Recall that our response vector $\bY$ is generated from the factor augmented linear model~\eqref{factor_aug2_s2}, but the latent variables $(\bF, \bU)$ are not observed.}
\fi

%%Estimation of the factor model~\eqref{eq_factor_model} has been well studied in the literature. In this section, we shall adopt the framework of constrained least squares estimation.% which is equivalent to the principal component analysis in the context of the factor model (\cite{Fan2013POET}).

Since only the predictor vector $\bx$ is observable, the latent factor $\bff$ and the corresponding loading matrix $\bB$ are not identifiable under the factor model~\eqref{eq_factor_model}. More specifically, for any non-singular matrix $\bS\in \bbR^{K \times K}$, we have $\bx = \bB\bff + \bu = (\bB \bS) (\bS^{-1}\bff) + \bu$. To resolve this issue, we impose the following identifiability conditions~\citep{Bai2003inferential, Fan2013POET}:
\begin{align*}
    \Cov(\bff) = \bI_{K} \enspace \mathrm{and} \enspace \bB^{\top} \bB \enspace \mathrm{is \ diagonal}.
\end{align*}
Consequently, the constrained least squares estimator of $(\bF, \bB)$ based on $\bX$ is given by
\begin{align*}%\label{eq_LSE_BF}
    (&\hat{\bF}, \hat{\bB}) = \underset{\bF \in \bbR^{n \times K}, \bB \in \bbR^{d \times K}}{\arg\min} \ \|\bX - \bF \bB^\top\|_{\bbF}^2\cr
    &\mathrm{subject\ to} \enspace \frac{1}{n} \bF^\top \bF = \bI_{K} \enspace \mathrm{and} \enspace \bB^\top \bB \enspace \mathrm{is \ diagonal}.
\end{align*}
Elementary manipulation yields that the columns of $\hat{\bF}/\sqrt{n}$ are the eigenvectors corresponding to the largest $K$ eigenvalues of the matrix $\bX \bX^{\top}$ and $\hat{\bB} = (\hat{\bF}^{\top}\hat{\bF})^{-1}\hat{\bF}^{\top}\bX = n^{-1}\hat{\bF}^{\top}\bX$. Then the least squares estimator for $\bU$ is given by $\hat{\bU} = \bX - \hat{\bF} \hat{\bB}^\top = (\bI_{n} - n^{-1}\hat{\bF}\hat{\bF}^{\top})\bX$.

Before presenting the asymptotic properties of the estimators $\{\hat{\bF}, \hat{\bB}, \hat{\bU}\}$, we first impose some regularity conditions.

\begin{assumption}
\label{Assumption_UF}
There exists a positive constant $c_{0} < \infty$ such that $\|\bff\|_{\psi_{2}} \leq c_{0}$ and $\|\bu\|_{\psi_{2}} \leq c_{0}$.
\end{assumption}

\begin{assumption}
\label{Assumption_factor}
There exists a constant $\tau > 1$ such that $d/\tau \leq \lambda_{\min}(\bB^{\top} \bB) \leq \lambda_{\max}(\bB^{\top} \bB) \leq d\tau$.
\end{assumption}

\begin{assumption}
\label{Assumption_moment}
Let $\bSigma = \Cov(\bu)$. There exists a constant $\Upsilon > 0$ such that $\|\bB\|_{\max} \leq \Upsilon$ and
\begin{align*}
    \bbE |\bu^{\top} \bu - \tr(\bSigma)|^{4} \leq \Upsilon d^{2}.
\end{align*}
\end{assumption}

\begin{assumption}
\label{Assumption_bSigma}
There exist a positive constant $\kappa < 1$ such that $\kappa \leq \lambda_{\min}(\bSigma)$, $\|\bSigma\|_{1} \leq 1/\kappa$ and $\min_{1\leq k, \ell \leq d} \Var(u_{k}u_{\ell}) \geq \kappa$.
\end{assumption}

\begin{remark}
Assumptions~\ref{Assumption_UF}--\ref{Assumption_bSigma} are standard assumptions in the studies of large dimensional factor model. We refer to~\citet{Bai2003inferential}, \citet{Fan2013POET} and~\citet{Li2018embracing} for more details.
%Assumption~\ref{Assumption_factor} pervasiveness condition. Assumption~\ref{Assumption_bSigma} ensures that the dependence is rather weak after removing the latent factors. imples that the first $K$ eigenvalues and the rest are separable, which lays the foundation of the PCA approach. We refer to \citet{Fan2013POET} for more details.
\qed
\end{remark}

We next summarize the theoretical results related to consistent factor estimation in the following proposition which directly follows from Lemmas D.1 and D.2 in~\citet{Wang2017}.

\begin{proposition}
\label{Lemma_factor_consistency}
Assume that $\log n = o(d)$. Let $\bH = n^{-1}\bV^{-1}\hat{\bF}^\top \bF \bB^\top \bB$, where $\bV \in \bbR^{K \times K}$ is a diagonal matrix consisting of the first $K$ largest eigenvalues of the matrix $n^{-1}\bX\bX^{\top}$. Then, under Assumptions~\ref{Assumption_UF}--\ref{Assumption_bSigma}, we have
\begin{enumerate}
    \item $\|\hat{\bF} - \bF \bH^\top\|_{\bbF}^{2} = O_{\bbP}(n/d + 1/n)$.
    \item For any $\mathcal{I} \subset \{1, 2, \ldots, d\}$, we have $\max_{\ell \in \mathcal{I}} \sum_{t = 1}^{n} |\hat{u}_{t\ell} - u_{t\ell}|^2 = O_{\bbP}(\log |\mathcal{I}| + n/d)$.
    \item $\|\bH^\top \bH - \bI_K\|_{\bbF}^{2} = O_{\bbP} (1/n + 1/d)$.
    \item $\max_{\ell \in [d]} \|\hat{\bb}_{\ell} - \bH\bb_{\ell}\|_{2}^{2} = O_{\bbP}\{(\log d)/n\}$.
\end{enumerate}
\end{proposition}

%\begin{remark}
%\textcolor{red}{Remark on the consistency of K.}
%\qed
%\end{remark}

\begin{remark}
In practice, the number of latent factors $K$ is typically unknown and it is an important issue to determine $K$ in a data-driven way. There have been various methods proposed in the literature to estimate the number $K$~\citep{Bai2002,Lam2012,Ahn2013, fan2021estimating}. Our theories always work as long as we replace $K$ by any consistent estimator $\hat{K}$, i.e. we only require
\begin{align*}
    \bbP(\hat{K} = K) \to 1, \enspace \mathrm{as} \enspace n \to \infty.
\end{align*}
Thus, without loss of generality, we assume the number of factors $K$ is known throughout all the theories  developed in this paper. As for the application part,  throughout this paper, we utilize the eigenvalue ratio method~\citep{Lam2012, Ahn2013} to select the number of factors. More specifically, we let $\lambda_{k}(\bX\bX^{\top})$ denote the eigenvalues of the Gram matrix $\bX\bX^{\top}$ and the number of factors is given by
\begin{align*}
    \hat{K} = \underset{K\leq \cK}{\arg\max} \frac{\lambda_{k}(\bX\bX^{\top})}{\lambda_{k + 1}(\bX\bX^{\top})},
\end{align*}
where $1\leq \cK \leq n$ is a prescribed upper bound for $K$.
\qed
\end{remark}

\subsection{Regularization Estimation}
\label{reg_est}
Under the high dimensional regime where the dimension $d$ can be much larger than the sample size $n$, it is often assumed that only a small portion of the predictors contribute to the response variable, which amounts to assuming that the true parameter vector $\bbeta^{\star}$ is sparse. 
Then the regularized estimator for the unknown parameter vectors $\bbeta^{\star}$ and $\bgamma^{\star}$ of our factor augmented linear model is defined as follows:
\begin{align}
\label{eq_Lasso}
    (\hbbeta_\lambda, \hat{\bgamma}) = \underset{\bbeta \in \bbR^{d}, \bgamma\in\bbR^{K}}{\arg\min} \left\{\frac{1}{2n}\|\bY - \hat{\bU} \bbeta - \hat{\bF} \bgamma\|_{2}^{2} + \lambda \|\bbeta\|_{1}\right\},
\end{align}
where $\lambda > 0$ is a tuning parameter.

We let $\tilde{\bY} = (\bI_{n} - \hat{\bP}) \bY$ denote the residuals of the response vector $\bY$ after projecting onto the column space of $\hat{\bF}$, where $\hat{\bP}  = n^{-1}\hat{\bF}\hat{\bF}^{\top}$ is the corresponding projection matrix. Recall that $\hat{\bU} = (\bI_{n} - \hat{\bP})\bX$. Hence $\hat{\bF}^\top \hat{\bU} = 0$ and it is straightforward to verify that the solution of~\eqref{eq_Lasso} is equivalent to
\begin{align*}
    \hat{\bbeta}_\lambda &= \underset{\bbeta\in\bbR^d}{\arg\min}\left\{\frac{1}{2n}\|\tilde{\bY} - \hat{\bU}\bbeta\|_{2}^{2} + \lambda \|\bbeta\|_1\right\}, \cr
    \hat{\bgamma} &= (\hat{\bF}^\top \hat{\bF})^{-1}\hat{\bF}^\top \bY =  \frac{1}{n}\hat{\bF}^\top \bY.
\end{align*}
%%%%%%%%%Comparing with the design matrix $\bX$ of which the columns are highly correlated, the dependence between $\hat{\bU}$ is much weaker. Observe that $\hat{\bgamma}$ coincides with the traditional PCR estimator based on the first $K$ principal components of the predictor. Intuitively, the regularized estimator $\hat{\bbeta}_{\lambda}$ quantifies the additional effect of the $\bu$.
For any subset $\cS$ of $\{1, \ldots, d\}$, we define the convex cone $\cC(\cS, 3) = \{\bdelta \in \bbR^{d} : \|\bdelta_{\cS^{c}}\|_{1} \leq 3 \|\bdelta_{\cS}\|_{1}\}$. For simplicity of notation, we write
\begin{align}
\label{def_cV}
    \cV_{n, d} = \frac{n}{d} + \sqrt{\frac{\log d}{n}} + \sqrt{\frac{n \log d}{d}}.
\end{align}
To investigate the consistency property of $(\hat{\bbeta}_{\lambda}, \hat{\bgamma})$, we impose the following moment condition on the random noise $\varepsilon$.

\begin{assumption}
\label{Assumption_e}
There exists a positive constant $c_{1} < \infty$ such that $\|\varepsilon\|_{\psi_{2}} \leq c_{1}$.
\end{assumption}

\begin{theorem}
\label{Theorem_beta}
Recall $\bvarphi^{\star} = \bgamma^{\star} - \bB^{\top}\bbeta^{\star} \in \bbR^{K}$. Under Assumptions~\ref{Assumption_UF}--\ref{Assumption_e}, we have
\begin{align*}
    \|\hat{\bgamma} - \bH \bgamma^{\star}\|_{2} = O_{\bbP}\left\{\frac{1}{\sqrt{n}} + \left(\frac{1}{\sqrt{n}} + \frac{1}{\sqrt{d}}\right)\|\bvarphi^{\star}\|_{2} + \|\bbeta^{\star}\|_{1} \left(\sqrt{\frac{\log |\cS_{\star}|}{n}} + \frac{1}{\sqrt{d}}\right)\right\},
\end{align*}
where $\cS_{\star} = \{j \in [d]: \beta_{j}^{\star} \neq 0\}$ and $|\cS_{\star}|$ is its cardinality. 
Furthermore, if
$
    |\cS_{\star}|\left(\frac{\log d}{n} + \frac{1}{d}\right) \to 0$,
then, by taking $\lambda = (\cI_{0}/n)\|\hat{\bU}^\top (\tilde{\bY} - \hat{\bU}\bbeta^{\star})\|_{\infty}$ for some constant $\cI_{0} \geq 2$, we have $\hbbeta_{\lambda} - \bbeta^{\star} \in \cC(\cS_{\star}, 3)$ and
\begin{align}
\label{eq_bound_bbeta}
    \|\hat{\bbeta}_{\lambda} - \bbeta^{\star}\|_{2} = O_{\bbP}\left(\sqrt{\frac{|\cS_{\star}|\log p}{n}} + \frac{\cV_{n, d}\|\bvarphi^{\star}\|_{2}\sqrt{|\cS_{\star}|}}{n}\right).
\end{align}
\end{theorem}

%\begin{remark}
%The first term of~\eqref{eq_bound_bbeta} is the usual LASSO bound. The second term quantifies that effect of the, which comes from the model missepcification error caused by replacing the unobserved $\bF$ by its empirical estimator $\hat{\bF}$. difference between the column space of $(\bF, \bU)$ and $\hat{\bF}, \hat{\bU}$. In the special case When $\bvarphi^{\star} = 0$,
%\end{remark}

\begin{remark}
In most of literature investigating the regularized estimation of sparse linear regression model~\eqref{eq4}, it is commonly assumed that the observed covariate vector $\bx$ is a sub-Gaussian random vector with bounded sub-Gaussian norm $\|\bx\|_{\psi_{2}}$. See, for instance, \citet{Loh2012}, \citet{Nickl2013confidence}, \citet{Sara2014}, \citet{Zhang2017} and many others. However, such assumption can be unreasonable in the presence of highly correlated covariates. To see this, suppose now both $\bff$ and $\bu$ are Gaussian random vectors and the underlying $\bx$ satisfies the factor model~\eqref{eq_factor_model}. Then $\bx$ is also a Gaussian random vector with $\Cov(\bx) = \bB\bB^{\top} + \bSigma$. Under the pervasiveness condition (Assumption~\ref{Assumption_factor}) and Assumption~\ref{Assumption_bSigma}, it is straightforward to verify that $\|\bx\|_{\psi_{2}} = \sqrt{8/3} \lambda_{\max}(\bB\bB^{\top} + \bSigma) \asymp d$, which violates the assumption on bounded sub-Gaussian norm. In contrast, our model can circumvent such issue because we decompose the covariate $\bx$ into $(\bff,\bu)$, and we only need impose sub-Gaussian assumption on $(\bff,\bu)$. As the sparse linear regression model serves as a special case to our model, our model serves as a more robust choice to conduct parameter estimation comparing with using linear regression directly, even if the sparse linear regression model is adequate.
\qed
\end{remark}

%The second term the effect of $(\bI_{n} - \hat{\bP})\bF$. Our results are also able to characterize how well we can estimate $\bbeta^{\star}$ even when the linear model $Y = \bx^{\top}\bbeta^{\star} + \varepsilon$ is not adequate. $\cV_{n, d}\|\bvarphi^{\star}\|_{2} \lesssim \sqrt{n \log d}$.

\begin{remark}
Theorem~\ref{Theorem_beta} substantially generalize the results in~\citet{Fan2020factor} with weaker assumptions. First, we did not impose the irrepresentable condition on the design matrix $\bU$, only the lower bound on $\bSigma = \Cov(\bu)$ is required. In addition, although~\citet{Fan2020factor} also decompose the covariate $\bx$ into $(\bff,\bu)$ in order to get precise estimator for $\bbeta^{\star}$, they mainly focus on the linear model $Y = \bx^{\top}\bbeta^{\star} + \varepsilon$ which corresponds to the special case with $\bvarphi^{\star} = 0$ in our results given in Theorem~\ref{Theorem_beta}.
\qed
\end{remark}

\begin{remark}
We note that our study is very different from the related work by~\cite{Fan2021_factor}, although they also study one kind of factor augment linear regression model. To be more specific, they assume the response $Y_{i, t}$ is given in a penal form with $i\in [N], t\in [T]$, which is generated from the model $Y_{i, t}=\blambda_i^\top \bff_t+u_{i, t}$. Here $u_{i, t}, i\in [N], t\in [T]$ is the idiosyncratic component. They incorporate the sparse linear regression into their model by assuming $u_{i,t}=\bbeta_i^\top\bu_{-i,t}+\epsilon_{i,t}$, $\forall i\in [N]$. Thus, their factor augmented sparse linear regression model heavily relies on the penal data structure.  In contrast, we study the cross-sectional data and focus on different inference problems.

%In our case, we only assume $Y_t,t\in [T]$ to be a time series.  \r{Q: brief discus} Meanwhile, when they only get access to the time series data (this is equivalent with setting $N=1$ in their model), such a linear relation between their $u_{1,t}$ and $\bu_{-1,t}$ doesn't exist and their model just reduces to latent factor regression.
\qed
\end{remark}

\iffalse
\newpage
\textcolor{blue}{
Possible Huber regression
\begin{align*}
    \hbbeta = \arg\min \sum_{t = 1}^{n} \rho_{\alpha}(\tilde{y}_{t} - \hat{\bu}_{t}^{\top}\bbeta) + \lambda \|\bbeta\|_{1}
\end{align*}
or
\begin{align*}
    (\hbbeta, \hbgamma) = \arg\min \sum_{t = 1}^{n} \rho_{\alpha} (y_{t} - \hat{\bu}_{t}^{\top}\bbeta - \hat{\bff}_{t}^{\top}\bgamma) + \lambda \|\bbeta\|_{1}.
\end{align*}
For the first approach, we need to bound
\begin{align*}
    \left\|\sum_{t = 1}^{n} \rho_{\alpha}' (\tilde{y}_{t} - \hat{\bu}_{t}^{\top}\bbeta^{\star})\right\|_{\infty}
\end{align*}
where
\begin{align*}
    \tilde{\bY} - \hat{\bU}\bbeta^{\star} = (\bI - \bP)\bF \bvarphi^{\star} + (\bI - \bP)\cE
\end{align*}
}
\textcolor{blue}{
Since $\bX \hat{\bB} = \frac{1}{n}\bX\bX^\top \hat{\bF} = \hat{\bF}\bV$, we have $\hat{\bF} = \bX \hat{\bB}\bV^{-1}$ and
\begin{align*}
    \bY \approx \hat{\bF} \bvarphi^{\star} + \bX \bbeta^{\star} + \cE = \bX(\hat{\bB}\bV^{-1}\bvarphi^{\star} + \bbeta^{\star}) + \cE
\end{align*}
Note that
\begin{align*}
    \|\hat{\bB}\bV^{-1}\bvarphi^{\star}\|_{\infty} \asymp \frac{\|\bvarphi^{\star}\|_{\infty}}{d}, \enspace \|\hat{\bB}\bV^{-1}\bvarphi^{\star}\|_{2} \asymp \frac{\|\bvarphi^{\star}\|_{2}}{\sqrt{d}} \enspace and \enspace  \|\hat{\bB}\bV^{-1}\bvarphi^{\star}\|_{1} \asymp \|\bvarphi^{\star}\|_{1}
\end{align*}
Intuitively, we can take
\begin{align*}
    \|\hbbeta_{0} - \hbbeta_{1}\|_{1}
\end{align*}
as test statistic
}
\fi

\subsection{Factor Augmented Robust Linear Regression}
In reality, datasets, especially collected from the field of finance, are often contaminated by noises with relatively heavy tails. To resolve such issue, we leverage the adaptive Huber regression to study the parameter of interest in our FARM under the existence of heavy-tailed noise~\citep{AHR2020}.

We first introduce some notation and basic definitions. Let $\rho_{\omega}(\cdot)$ denote the Huber function,
\begin{align*}
    \rho_{\omega}(z) = \left\{
    \begin{array}{cc}
        z^{2}/2, & \mathrm{if} \enspace |z|\leq \omega,  \\
        \omega z - \omega^{2}/2, & \mathrm{if} \enspace |z| > \omega,
    \end{array}
    \right.
\end{align*}
where $\omega > 0$ is the robustification parameter which balances robustness and bias. Following the intuition of~\eqref{eq_Lasso}, our factor augmented Huber estimator for $(\bbeta^{\star}, \bgamma^{\star})$ is given by
\begin{align}
\label{robust_regression}
    (\hbbeta_{h}, \hbgamma_{h}) = \underset{\bbeta\in\bbR^{d}, \bgamma \in \bbR^{K}}{\arg\min}\left\{\frac{1}{n}\sum_{t = 1}^{n} \rho_{\omega} (y_{t} - \hat{\bu}_{t}^{\top}\bbeta - \hat{\bff}_{t}^{\top}\bgamma) + \lambda \|\bbeta\|_{1}\right\},
\end{align}
where $\lambda > 0$ is a tuning parameter. For simplicity of notation, we write $\hat{\bphi}_{h} = (\hbbeta_{h}^{\top}, \hbgamma_{h}^{\top})^{\top} \in \bbR^{d + K}$ and $\tilde{\bphi} = (\bbeta^{\star \top}, \tilde{\bgamma}^{\top})^{\top}\in \bbR^{d + K}$, where $\tilde{\bgamma} = \hat{\bB}^{\top} \bbeta^{\star} + n^{-1} \hat{\bF}^{\top} \bF\bvarphi^{\star}$. The following theorem establishes the statistical consistency of $\hat{\bphi}_{h}$.

\begin{proposition}
\label{Proposition_Huber_regression}
Assume that $\bbE |\varepsilon|^{1 + \vartheta} < \infty$ for some constant $\vartheta > 0$. Let
\begin{align*}
    \omega \asymp \left(\frac{n}{\log d}\right)^{\frac{1}{1 + (\vartheta \wedge 1)}} \enspace \mathrm{and} \enspace \lambda \asymp \left(\frac{\log d}{n}\right)^{\frac{\vartheta \wedge 1}{1 + (\vartheta \wedge 1)}}.
\end{align*}
Furthermore, we assume that $(|\cS_{\star}| + K)(\log d)^{3/2} = o(n)$,
\begin{align}
\label{eq_Huber_cond_varphi}
    \frac{\log n}{n + \sqrt{d}} \|\bvarphi^{\star}\|_{2} = o(\omega) \enspace \mathrm{and} \enspace \cV_{n, d} \|\bvarphi^{\star}\|_{2} = O(\omega \log d).
\end{align}
Then, under Assumptions~\ref{Assumption_UF}--\ref{Assumption_bSigma}, we have
\begin{align*}
    \|\hat{\bphi}_{h} - \tilde{\bphi}\|_{1} = O_{\bbP}\left\{(|\cS_{\star}| + K)\left(\frac{\log d}{n}\right)^{\frac{\vartheta \wedge 1}{1 + (\vartheta \wedge 1)}}\right\}.
\end{align*}
\end{proposition}

We establish the $\ell_{1}$-statistical rate for our parameters in model~\eqref{factor_aug}[also~\eqref{eq_model_varphi}] by only assuming the existence of $(1+\vartheta)$-th moment of the noise distribution. Specifically, when $\vartheta \geq 1$, the results reduce to the same rates as the sub-Gaussian assumption of $\varepsilon$.  Our result serves as an extension of~\cite{AHR2020} to a more general setting by incorporating latent factors.

\section{Is Factor Regression Model Adequate?}
\label{pcr_test}
The latent factor regression is widely applied in many fields as an efficient dimension reduction method.
%%%%%%%%% However, the PCR relies on the strict assumption that the response depends only on the first several principal components. In spite of its compact expression and nice statistical properties, the efficiency and adequacy of the PCR model remains open. As alternative, several data-adaptive principal components regression have been proposed. A more important area concerning the adequacy and efficiency of PCR is less studied.
%In spite of its compact expression and, or need to fully utilize the information from observed covariate $\bx$. Moreover, the PCR model relies on the assumption that the only the first several principal components contribute to the response.
A natural question arises is whether the model is adequate and FARM \eqref{factor_aug} serves naturally as the alternative model.  To be more specific, we consider testing the hypotheses
\begin{align}
\label{eq_PCR_test}
    H_0 : \bbeta^{\star} = 0 \enspace \mathrm{versus} \enspace H_1 : \bbeta^{\star} \neq 0
\end{align}
in FARM \eqref{factor_aug}. As the penalized least-squares estimator $\hat\bbeta_{\lambda}$ is used for estimating $\bbeta^{\star}$, it creates biases and make it difficulty for inferences. Thus,  we first introduce a de-biased version of $\hat\bbeta_{\lambda}$ given in \eqref{eq_Lasso}.

%\textcolor{red}{reference on global hypothesis in high dimensional linear regression}

\subsection{Bias Correction}
We begin with the construction of bias-corrected estimator for $\bbeta^{\star}$ following similar idea of~\citet{Zhang2014}, \citet{Sara2014} and~\citet{Javanmard2014}. Specifically, let $\hat{\bTheta} \in \bbR^{d \times d}$ be an approximation for the inverse of the Gram matrix $\tilde{\bSigma} = n^{-1} \hat{\bU}^\top \hat{\bU}$, the de-biased estimator for $\bbeta^{\star}$ is then defined as
\begin{align}\label{eq_Debiased}
    \tilde{\bbeta}_{\lambda} = \hbbeta_{\lambda} + \frac{1}{n}\hat{\bTheta} \hat{\bU}^\top (\bY - \hat{\bU} \hbbeta_{\lambda}).
\end{align}
The rationale behind such construction is that we are able to decompose estimation error as
\begin{align}\label{eq_debiased_beta}
    \tilde{\bbeta}_{\lambda} - \bbeta^{\star} = \frac{1}{n} \hat{\bTheta}\hat{\bU}^\top \cE + \frac{1}{n}\hat{\bTheta}\hat{\bU}^\top \bF \bvarphi^{\star} + (\bI_d - \hat{\bTheta}\tilde{\bSigma})(\hbbeta_{\lambda} - \bbeta^{\star}),
\end{align}
after we expand $\bY$ according to~\eqref{eq_model_matrix} and replace $\bX$ by $\bX=\hat\bF\hat\bB+\hat\bU$. The first term on the right hand side of \eqref{eq_debiased_beta} quantifies the uncertainty of our estimator $\tilde{\bbeta}_{\lambda}$ and the last two terms are biases which will be shown to be of smaller order.

One observes that constructing the de-biased estimator $\tilde{\bbeta}_{\lambda}$ given above requires an estimator $\hat\bTheta$. There exist many methodologies devoting to estimating such precision matrix, for example, the node-wise regression proposed in~\citet{Zhang2014} and~\citet{Sara2014}, and the CLIME-type estimator given in~\citet{Cai2011}, ~\citet{Javanmard2014} and ~\citet{avella2018robust}. In our work, we do not restrict $\hat{\bTheta}$ to be any specific one, but require to satisfy the following general conditions.

\begin{assumption}
\label{Assumption_tilde_Theta}
Let $\bTheta = \bSigma^{-1}$ with $\bSigma$ defined in Assumption \ref{Assumption_moment}. There exist positive $\Lambda_{\max}$ and $\Delta_{\infty}$ such that
\begin{align*}
    \|\bI_d - \hat{\bTheta}\tilde{\bSigma}\|_{\max} = O_{\bbP}(\Lambda_{\max}) \enspace \mathrm{and} \enspace \|\hat{\bTheta} - \bTheta\|_{\infty} = O_{\bbP}(\Delta_{\infty}).
\end{align*}
\end{assumption}
Without loss of generality, here we assume that $\Delta_{\infty} \leq \|\bTheta\|_{\infty}$.

\iffalse
\begin{remark}
In Appendix~\ref{example_node}, we provide a concrete example of how to estimate $\tilde{\bSigma}^{-1}$ via node-wise regression~\citep{Zhang2014, Sara2014}. Under the mild conditions therein, we establish that Assumption~\ref{Assumption_tilde_Theta} is satisfied with
\begin{align*}
    \Lambda_{\max} \asymp \sqrt{\frac{\log d}{n}} + \frac{1}{\sqrt{d}} \enspace \mathrm{and} \enspace \Delta_{\infty} \asymp \max_{j \in [d]}|\cS_{j}|\sqrt{\frac{\log d}{n} + \frac{1}{d}},
\end{align*}
where $|\cS_{j}| = \sum_{k = 1}^{d} \mathbb{I}\{\Theta_{kj} \neq 0\}$ measures the sparseness of $j$-th column of the precision matrix $\bTheta$ for each $1\leq j\leq d$.
\qed
\end{remark}
\fi

\begin{remark}
To give a concrete example, under the mild conditions therein, Assumption~\ref{Assumption_tilde_Theta} is satisfied with
\begin{align*}
    \Lambda_{\max} = O\bigg(\sqrt{\frac{\log d}{n}} + \frac{1}{\sqrt{d}}\bigg) \enspace \mathrm{and} \enspace \Delta_{\infty} = O\bigg(\max_{j \in [d]} |\cS_{j}|\sqrt{\frac{\log d}{n} + \frac{1}{d}}\bigg),
\end{align*}
by using node-wise regression~\citep{Zhang2014, Sara2014}, where $|\cS_{j}| = \sum_{k = 1}^{d} \mathbb{I}\{\Theta_{jk} \neq 0\}$ quantifies the sparsity of $j$-th column of the precision matrix $\bTheta$ for each $1\leq j\leq d$. In Appendix \textcolor{red}{C.1}, we will provide a detailed analysis on estimating $\tilde{\bSigma}^{-1}$ via node-wise regression and establish precise theoretical upper bounds for the statistical rates given in Assumption~\ref{Assumption_tilde_Theta}.
\qed
\end{remark}

\subsection{Gaussian Approximation}
The goal of this section is to derive the asymptotic distribution of $\|\tilde{\bbeta}_{\lambda} - \bbeta^{\star}\|_{\infty}$ in the high dimensional setting. To this end, we apply the Gaussian approximation result given in~\cite{CCK2013, CCK2017, CCK2020} for high dimensional random vectors. More specifically, we let $\bZ = (Z_{1}, \ldots, Z_{d})^{\top} \in \bbR^{d}$ be a zero-mean Gaussian random vector with the same covariance matrix as that of $n^{-1/2} \bTheta \bU^\top \cE$, that is,
\begin{align}
\label{eq_Cov_Z}
    \Cov(\bZ) = \Cov\left(\frac{1}{\sqrt{n}}\bTheta\bU^\top \mathcal{E}\right) = \sigma^2 \bTheta.
\end{align}
We next present the theoretical results on Gaussian approximation of our test statistics under some mild conditions.

\begin{theorem}
\label{Theorem_GA}
Recall $\bvarphi^{\star} = \bgamma^{\star} - \bB^{\top}\bbeta^{\star} \in \bbR^{K}$. We assume that $(\log d)^{5}/n \to 0$,
\begin{align}
\label{eq_GA_cond}
    (\Lambda_{\max} |\cS_{\star}| + \Delta_{\infty}) \log d \to 0 \enspace \mathrm{and} \enspace \left(\cV_{n, d} \|\bvarphi^{\star}\|_{2} + \sqrt{\frac{n}{d}} + \sqrt{\log d}\right)\|\bTheta\|_{\infty} \sqrt{\frac{\log d}{n}} \to 0,
\end{align}
with $\cV_{n, d}$ given by~\eqref{def_cV}. Then under Assumption~\ref{Assumption_tilde_Theta}, we have
\begin{align*}
    \sup_{x > 0}\left|\bbP\left(\sqrt{n}\|\tilde{\bbeta}_\lambda - \bbeta^{\star}\|_\infty \leq x\right) - \bbP\left(\|\bZ\|_\infty \leq x\right)\right| \to 0.
\end{align*}
\end{theorem}

For any $\alpha\in(0, 1)$, let $c_{1 - \alpha}$
% = \inf\{t \geq 0: \bbP(\|\bZ\|_\infty\leq t) \geq 1 - \alpha\}$ 
denote the $(1 - \alpha)$-th quantile of the distribution of $\|\bZ\|_\infty$.  Theorem~\ref{Theorem_GA} leads to an approximately level $\alpha$ test for~\eqref{eq_PCR_test} as follows:
\begin{align}
\label{eq_test_psi_infty}
    \psi_{\infty, \alpha} = \mathbb{I}\left\{\sqrt{n}\|\tilde{\bbeta}_\lambda\|_\infty > c_{1 - \alpha}\right\}.
\end{align}
%Given the value of $\psi_{\infty, \alpha} \in \{0, 1\}$, we shall reject the null hypothesis~\eqref{eq_PCR_test} if and only if $\psi_{\infty, \alpha} = 1$.

%The proposed test $\psi_{\infty, \alpha}$ in~\eqref{eq_test_psi_infty} is infeasible in practice since the critical value $c_{1 - \alpha}$ depends on the unknown $\sigma^2$ and $\bTheta$. In order to proceed our testing procedure, in the next section, we propose to estimate $c_{1 - \alpha}$ by plugging in sample estimators $\hat{\bTheta}$ and $\hat{\sigma}$ and conducting a Gaussian multiplier bootstrap procedure.

\subsection{Gaussian multiplier bootstrap}
\label{bootstrap}
The critical value $c_{1 - \alpha}$ depends on the unknown $\sigma^2$ and $\bTheta$, which can be estimated by the following Gaussian multiplier bootstrap.  
\begin{enumerate}
    \item Generate i.i.d.~random variables $\xi_{1}, \ldots, \xi_{n} \sim N(0, 1)$ and compute
    \begin{align*}
        \hat{L} = \frac{1}{\sqrt{n}}\|\hat{\bTheta}\hat{\bU}^\top\bxi\|_\infty,  \enspace \mathrm{where} \enspace \bxi = (\xi_1, \xi_2, \ldots, \xi_n)^\top.
    \end{align*}
    \item Repeat the first step independently for $B$ times and obtain $\hat{L}_{1}, \ldots, \hat{L}_{B}$. Estimate the critical value $c_{1 - \alpha}$ via $1-\alpha$ quantile of the empirical distribution of the bootstrap statistics:
    \begin{align*}
        \hat{c}_{1 - \alpha} = \inf\{t \geq 0 : H_{B}(t) \geq 1 - \alpha\}, \enspace \mathrm{where} \enspace H_{B}(t) = \frac{1}{B} \sum_{b = 1}^{B} \mathbb{I}\left\{\hat{L}_{b} \leq t\right\}.
    \end{align*}
\end{enumerate}
Reject the null hypothesis $H_0$ when $\sqrt{n}\|\tilde{\bbeta}_{\lambda}\|_{\infty}/\hat\sigma > \hat{c}_{1 - \alpha}$, for a given consistent estimator
$\hat{\sigma}$ of $\sigma$. To validate the procedure, we need some additional conditions on $\hat{\bTheta}$ and $\hat{\sigma}$.

\begin{assumption}
\label{Assumption_bTheta_max}
	There exists a $\Delta_{\max} > 0$ such that $\|\hat{\bTheta} - \bTheta\|_{\max} = O_{\bbP}(\Delta_{\max})$.
\end{assumption}

\begin{assumption}
\label{Assumption_noise}
	There exists a $0 < \Delta_\sigma \leq 1$ such that $|\hat{\sigma}/\sigma - 1| = O_{\bbP}(\Delta_{\sigma})$.
\end{assumption}

\begin{remark}
The estimation of $\sigma^{2}$ for high dimensional linear regression has been extensively in the literature. For example, \citet{Fan2012variance} proposed refitted cross-validation to construct a consistent estimator with clearly quantified uncertainty of $\hat{\sigma}$ in ultra-high dimension. In addition, \citet{Sun2012scaled} and~\citet{Yu2019estimating} derived scaled-Lasso and organic Lasso respectively for estimating $\sigma$. Like our case of estimating $\bTheta$, we also do not restrict estimating $\sigma$ by any fixed method mentioned above, our theory works as long as the general condition of Assumption~\ref{Assumption_noise} holds.
%Please also see \textcolor{red}{(cite)} on estimating $\sigma$ for more details.
\qed
\end{remark}

%Given estimator $\hat\bTheta$ satisfying Assumptions~\ref{Assumption_tilde_Theta} and~\ref{Assumption_bTheta_max}, and $\hat{\sigma}$ satisfying Assumption~\ref{Assumption_noise}, we then further propose a Gaussian multiplier bootstrap procedure to estimate the critical value $c_{1 - \alpha}$. The relevant steps are described below.

Let $\bbP^{\star}(\cdot) = \bbP(\cdot | \bX, \bY)$ denote the conditional probability. In the following theorem, we establish the validity of the proposed bootstrap procedure.
\begin{theorem}
\label{Theorem_GMB}
Let Assumptions~\ref{Assumption_tilde_Theta}--\ref{Assumption_noise} hold. Assume that
\begin{align}
\label{eq_cond_GMB}
    \Lambda_{\max} \|\bTheta\|_{\infty} + \Delta_{\max}+\Delta_{\sigma} = o\left(\frac{1}{\log d}\right).
\end{align}
Then, under conditions of Theorem~\ref{Theorem_GA}, we have
\begin{align*}
    \sup_{x > 0} \left|\bbP\left(\sqrt{n}\|\tilde{\bbeta}_\lambda - \bbeta^{\star}\|_\infty \leq x\right) - \bbP^{\star}\left(\hat{L} \leq x\right)\right| \overset{\bbP}{\to} 0.
\end{align*}
\end{theorem}

\begin{remark}
%In the section given above, we aim at constructing simultaneous confidence intervals for components $\bbeta^{\star}$. It is worth to note that 
Following the same de-biasing procedure as given in~\eqref{eq_Debiased}, we are also able to construct entrywise~\citep{Javanmard2014} and groupwise~\citep{Zhang2017, Dezeure2017} simultaneous confidence intervals for $\bbeta^{\star}$.  For each $1\leq j\leq d$, a $(1-\alpha)$-confidence interval for $\beta_{j}^{\star}$ is given by
\begin{align*}
    \mathcal{CI}_{\alpha}(\beta_{j}^{\star}) = \left\{\tilde{\beta}_{j, \lambda} - \hat{\sigma} z_{1 - \alpha/2}\sqrt{\frac{\hat{\Theta}_{jj}}{n}}, \  \tilde{\beta}_{j, \lambda} - \hat{\sigma} z_{1 - \alpha/2}\sqrt{\frac{\hat{\Theta}_{jj}}{n}}\right\},
\end{align*}
where $z_{1 - \alpha/2}$ is the $(1 - \alpha/2)$-th quantile of standard normal distribution. For simultaneous groupwise inference of $\bbeta^{\star}$, let $G$ be a subset of $\{1, \ldots, d\}$  of interest and consider testing the hypotheses
\begin{align*}
    H_{0, G} : \beta_{j}^{\star} = \beta_{j}^{\circ} \enspace \mathrm{for \ all } \enspace j \in G \enspace \mathrm{versus} \enspace H_{1, G} : \beta_{j}^{\star} \neq \beta_{j}^{\circ} \enspace \mathrm{for \ some} \enspace j \in G.
\end{align*}
In particular, when $\beta_{j}^{\circ} = 0$ for all $j \in G$, this reduces to testing the significance of a group of parameters. We obtain that the asymptotic distribution of $\max_{j \in G} \sqrt{n} |\tilde{\beta}_{j,\lambda} - \beta_{j}^{\star}|$ converges to the distribution of $\max_{j \in G} |Z_{j}|$ by leveraging the Gaussian approximation. The remaining steps follow directly by conducting the Gaussian multiplier bootstrap.
\qed
\end{remark}

\iffalse
\textcolor{red}{Discuss on other related literature. Such as Cheng et al.
Possible directions.
\begin{enumerate}
	\item Clarify the assumptions.
	%\item Compare with other methods
	%\item Power analysis.
	\item applications: support recovery? FDR control?
	\item generalization: generalized linear model?
\end{enumerate}
}
\fi

\section{Is Sparse Linear Model Adequate?}
\label{sparse_test}
Sparse linear regression, which serves as the backbone of high dimensional statistics, has been widely applied in many areas of science, engineering, and social sciences. However, its adequacy has never been validated.  This section focuses on testing the adequacy of the sparse linear model. %Its theoretical properties in terms of the statistical estimation have been well investigated by \cite{LASSO1996, SCAD2001,ALASSO2006, CandesTao2007, Bickel2009, MCP2010}. %Please also see~\citet{Fan_Li_Zhang_Zou_data} for a comprehensive survey.

\subsection{Main Results}
As mentioned in introduction, the proposed model~\eqref{eq_model_varphi} contains the sparse linear regression model as a special case. Thus,  we consider testing the hypotheses
\begin{align}
\label{Test_sparse_model}
    H_{0} : Y = \bx^{\top} \bbeta^{\star} + \varepsilon \enspace \mathrm{versus} \enspace H_{1} : Y = \bff^{\top}\bvarphi^{\star} + \bx^{\top}\bbeta^{\star} + \varepsilon,
\end{align}
which is equivalent to test whether $\bvarphi^{\star}= \bgamma^{\star} - \bB^{\top} \bbeta^{\star}= 0$. Since $\bB$ is an unknown dense matrix, simultaneously testing this linear equation will suffer from the curse of dimensionality.

On the other hand,  for any set $\cS \subset [d]$ with $\cS_{\star} \subset \cS$, we have $\bB_{\cS}^\top \bbeta_{\cS}^{\star} = \bB^\top \bbeta^{\star}$. Hence, it suffices to compare the following two linear models in reduced dimension:
\begin{align}
\label{Test_XY}
    H_{0} : Y = \bx_{\cS}^{\top} \bbeta_{\cS}^{\star} + \varepsilon \enspace \mathrm{versus} \enspace H_{1} : Y = \bff^{\top} \bvarphi^{\star} + \bx_{\cS}^{\top} \bbeta_{\cS}^{\star} + \varepsilon.
\end{align}
This hinges applying a sure screening method to reduce the dimensionality.  There exist several methods which lead to the sure screening property.
%See \citet{Fan2008sure,FanSong2010,Wang2012} and~\citet{Fan2020factor} for more details.
Among those, the commonly used one is the marginal screening method~\citep{Fan2008sure,  FanSong2010, %Fan2011nonparametric,
Zhu2011model, Li2012robust, %Wang2012, Li2012feature, 
Liu2014feature, Barut2016conditional, Chu2016feature, Wang2016high}. 

We propose an ANOVA-type test for \eqref{Test_sparse_model} with two stages. In the first stage, the data set is split into two data sets $(\bY^{(1)}, \bX^{(1)})$ and $(\bY^{(2)}, \bX^{(2)})$, with sample sizes $m$ and $n - m$, respectively. We use $(\bY^{(1)}, \bX^{(1)})$ to screen variables. Let $\hat{\cS}_1$ denote the set of variables selected. In the second stage, we leverage the selected $\hat{\cS}_1$ and remaining data $(\bY^{(2)}, \bX^{(2)})$ to perform hypothesis testing based on the ANOVA-type test statistic for low-dimensional model \eqref{Test_XY} with $\cS$ replaced by $\hat \cS_1$.  As the first step is based on marginal screening and is relatively crude, the sample size $m$ is relatively small in comparing with the second step.  We impose a general assumption on the set $\hat S_1$.
\begin{assumption}[Sure screening property]
\label{Assumption_sure_screening}
There exists a $s_n > 0$ such that
\begin{align*}
    \bbP\left(|\hat{\cS}_{1}| \leq s_{n} \enspace \mathrm{and} \enspace \cS_{\star} \subset \hat{\cS}_1\right) \to 1, \enspace \mathrm{as} \enspace n \to \infty.
\end{align*}
\end{assumption}

A simple procedure that satisfies the above assumption is the follow factor-adjusted marginal screening based on the data
$(\bY^{(1)}, \bX^{(1)})$.
\begin{enumerate}
    \item Estimation. Compute the latent factor estimator $\hat{\bF}^{(1)}$,  idiosyncratic component $\hat\bU^{(1)}$ based on $\bX^{(1)}$, and $\tilde{\bY}^{(1)}=(\bI_{m}-\hat \bF^{(1)}(\hat\bF^{(1)\top}\hat\bF^{(1)})^{-1}\hat\bF^{(1)\top})\bY^{(1)}$.
    \item Marginal regression. Compute the least square estimate $\hat{\beta}_{\ell, M} = \hat{\bU}_{\ell}^{(1)\top} \tilde{\bY}^{(1)} /(\hat{\bU}_{\ell}^{(1)\top} \hat{\bU}_{\ell}^{(1)})$ for each $1\leq \ell \leq d$.
   \item Screening. Let $\hat S_1:=\hat{\cS}_{\phi} = \{\ell \in [d]: |\hat{\beta}_{\ell, M}| > \phi\}$ for some prescribed $\phi > 0$.
\end{enumerate}
Here $\hat{\bU}_{\ell}^{(1)} \in \bbR^{d}$ stands for the $\ell$-th column of the matrix $\hat{\bU}^{(1)}$. We next provide a sufficient condition for the Assumption~\ref{Assumption_sure_screening} to hold.

\begin{proposition}
\label{Lemma_Sure_Screening}
Assume that $m = o(d\log d)$ and there exists some positive constant $\bar{c} < 1$ such that
\begin{align}
\label{eq_Sure_Screening_cond}
    \phi < \frac{1}{1 + \bar{c}}\min_{\ell \in [d]}\frac{\bSigma_{\ell}^{\top} \bbeta^{\star}}{\Sigma_{\ell\ell}} \enspace \mathrm{and} \enspace \frac{\cV_{m, d}}{m} \|\bvarphi^{\star}\|_{2} + \|\bbeta^{\star}\|_{1} \frac{\log d}{m} + \|\bbeta^{\star}\|_{2} \sqrt{\frac{\log d}{m}} = o(\phi).
\end{align}
Here $\bSigma_{\ell}$ denotes the $\ell$-th column of $\bSigma$. Then, under the Assumptions~\ref{Assumption_UF}--\ref{Assumption_e}, we have
\begin{align*}
    \bbP\left(\cS_{\star} \subset \hat{\cS}_{\phi}\right) \to 1, \enspace \mathrm{as} \enspace m \to \infty.
\end{align*}
Furthermore, we assume that $\min_{\ell \in \cS_{\star}} |\beta_{\ell, M}^{\star}:={\bSigma_{\ell}^{\top} \bbeta^{\star}}/{\Sigma_{\ell\ell}} | \geq c_{\star} m^{-\kappa}$ for some positive constant $\kappa < 1/2$. Then for any $\phi = c_{\diamond} m^{-\kappa}$ with $c_{\diamond}\leq c_{\star}/(1 + \bar{c})$, we have
\begin{align*}
    \bbP\left\{|\hat{\cS}_{\phi}| \leq \frac{c_{\diamond}^{2} m^{2\kappa}\|\bSigma\bbeta^{\star}\|_{2}^{2}}{\lambda_{\min}^{2}(\bSigma)(1 - \bar{c})^{2}}\right\} \to 1 \enspace \mathrm{as} \enspace m \to \infty.
\end{align*}
\end{proposition}

\begin{remark}
From the conclusion of Proposition~\ref{Lemma_Sure_Screening}, we obtain sure screening property by using our first data set with sample size $m=n^{\alpha}$ for some $\alpha<1$ as long as the signal satisfies $\min_{\ell \in \cS_{\star}} |\beta_{\ell, M}^{\star}| \geq c_{\star} m^{-\kappa}$. Thus, the size of the remaining data set for constructing the test statistic in our second step is $n - n^{\alpha}\approx n$. It is worth to note that this does not lose any efficiency in terms of the asymptotic power in our hypothesis test when $n$ goes to infinity.
\qed
\end{remark}

\begin{remark}
\citet{Fan2020factor} proposed a similar sure screening estimator which is a special case of our Proposition~\ref{Lemma_Sure_Screening} with $\bvarphi^{\star}= \bgamma^{\star} - \bB^{\top} \bbeta^{\star}= 0$.  Moreover, we also provide an upper bound for the number of selected variables  whereas \citet{Fan2020factor} only provided a sufficient condition for the sure screening property.
\qed
\end{remark}

%Here we do not dig details of Assumption \ref{Assumption_sure_screening} since it is not the main focus of the paper. Please kindly refer to Theorem 4.2 in~\cite{Fan2020factor}, which provides sufficient conditions to ensure Assumption~\ref{Assumption_sure_screening}, for more technical details. The optimal choice of $\phi$ is discussed in \textcolor{red}{(cite)}(Fan and Song). Here we do not dig details of this result since it is not the main focus of the paper.
%Please kindly refer to the references discussed above on the sufficient condition.
%\qed

Next, we proceed to the second stage of our hypothesis testing. In this step,  we construct an ANOVA test statistic for (\ref{Test_XY}) with $\cS$ replaced by $\hat \cS_1$, which is given by
\begin{align}
\label{Q_n2}
    Q_n^{(2)} = \left\|\left(\bI_{n-m} - \bP_{\bX_{\hat{\cS}_1}^{(2)}}\right)\bY^{(2)}\right\|_2^2 - \left\|\left(\bI_{n-m} - \bP_{\hat{\bF}^{(2)}} -  \bP_{\hat{\bU}_{\hat{\cS}_1}^{(2)}}\right)\bY^{(2)}\right\|_2^2.
\end{align}
We then summarize our results on the asymptotic behaviors of $Q_n^{(2)}$ in the following Theorem~\ref{Theorem_Chi_test}.

\begin{theorem}\label{Theorem_Chi_test}
Let Assumptions~\ref{Assumption_UF}--\ref{Assumption_e} and Assumption~\ref{Assumption_sure_screening} hold with
\begin{align*}
    s_n\left(\frac{\log d}{n} + \frac{1}{d}\right) \to 0 \enspace \mathrm{and} \enspace \Delta_{\sigma} \to 0.
\end{align*}
We obtain
\begin{align*}
    \sup_{x > 0} \left|\bbP\left({Q_{n}^{(2)}} \leq x {\hat{\sigma}^2} | H_{0}\right) - \bbP(\chi_{K}^2 \leq x)\right| \to 0, \enspace \mathrm{as} \enspace n \to \infty.
\end{align*}
\end{theorem}

Theorem~\ref{Theorem_Chi_test} yields a level $\alpha$ test for~\eqref{Test_sparse_model} with critical region
$\left\{Q_n^{(2)} > \hat{\sigma}^2 \chi_{K, 1-\alpha}^2\right\}$,
where $\chi_{K, 1-\alpha}^2$ is the $(1-\alpha)$-th quantile of $\chi_K^2$-distribution. 
%Given the value of $\psi_{\alpha} \in \{0, 1\}$, we shall reject the null hypothesis whenever $\psi_{\alpha} = 1$. We then obtain
%\begin{align*}
%    \sup_{\alpha \in (0, 1)} |\bbP(\psi_{\alpha} = 1 | H_{0}) - \alpha| \to 0, \enspace \mathrm{as} \enspace n \to \infty.
%\end{align*}

\begin{remark}
Under stronger conditions such as irrepresentable condition~\citep{ZhaoYu2006} or RIP condition~\citep{CandesTao2007}, the $\hat{\cS}$ achieved by certain explicit regularization~\citep{ZhaoYu2006, FanLv2011, Shi2019, Fan2020factor} or implicit regularization accompanied with early stopping and signal truncation~\citep{Zhao2019, Yu2021} enjoys variable selection consistency $\bbP(\hat{\cS} = \cS_{\star}) \to 1$. In this scenario, we take the test statistic as
\begin{align*}
    Q_{n} = \left\|\left(\bP_{\hat{\bF}} + \bP_{\hat{\bU}_{\hat{\cS}}} - \bP_{\bX_{\hat{\cS}}}\right)\bY\right\|_2^{2}
\end{align*}
without using sample splitting.
Under Assumptions~\ref{Assumption_UF}--\ref{Assumption_e}, we obtain
\begin{align}
\label{eq_Qn_distribution}
    \sup_{x > 0} \left|\bbP\left({Q_{n}} \leq x {\hat{\sigma}^{2}} | H_{0}\right) - \bbP(\chi_{K}^{2} \leq x)\right| \to 0,
\end{align}
by following similar proof idea with Theorem \ref{Theorem_Chi_test}.
\iffalse
The proof idea of \eqref{eq_Qn_distribution} is straightforward.
We let
\begin{align*}
    Q_{n}^{\star} = \left\|\left(\bP_{\hat{\bF}} + \bP_{\hat{\bU}_{\cS^{\star}}} - \bP_{\bX_{\cS^{\star}}}\right)\bY\right\|^{2},
\end{align*}
which satisfies the asymptotic behavior given in Theorem \ref{Theorem_Chi_test}. As $\bbP(\hat{\cS} = \cS^{\star}) \to 1$, it directly follows that
\begin{align*}
    \sup_{x > 0} |\bbP(Q_{n} \leq x) - \bbP(Q_{n}^{\star} \leq x)| \leq 2 \bbP(\hat{\cS} \neq \cS^{\star}) \to 0.
\end{align*}
\fi
\qed
\end{remark}

%\textbf{\textcolor{red}{Discuss on other related literatures, the paper of Quefengli?}}
\iffalse
\begin{enumerate}
	\item clarify assumptions
	\item compare with other results, e.g. the decorrelated score test doesn't work; If we have consistent estimation result, does other methods also hold?
	\item generalization? multimodal?
	\item power of the test
\end{enumerate}
\fi

We now present  the power of the test statistic~\eqref{Q_n2}.

\begin{theorem}
\label{Theorem_ANOVA_power}
Define
\begin{align*}
    \cD(\alpha, \theta) = \left\{\bvarphi \in \bbR^{K}: \frac{n\|\bvarphi\|^{2}}{1 + K s_{n} \|\bB\|_{\max}^{2}/\lambda_{\min}(\bSigma)} \geq \sigma^{2} (2 + \delta) (\chi_{K, 1 - \alpha}^{2} + \chi_{K, 1 - \theta}^{2})\right\},
\end{align*}
where $\delta > 0$ is some constant, $s_n$ is the size of selected set from the first stage and $K$ is the number of factors. Assume that 
\begin{align}\label{eq_power_cond}
\|\bvarphi^{\star}\|_{2} \left(\sqrt{n/d} + 1/\sqrt{n}\right) \to 0. 
\end{align}
  Then, under the conditions of Theorem~\ref{Theorem_Chi_test}, we have
\begin{align*}
    \inf_{\bvarphi^{\star} \in \cD(\alpha, \theta)} \bbP(\psi_{\alpha} = 1 | H_{1}) \geq 1 - \theta.
\end{align*}
\end{theorem}

%\begin{remark}
%Given $\alpha$ and $\theta$, the parameter space %$\cD(\alpha,\theta)$ is broad which also contains %weak signal $\bvarphi$ as long as $\|\bvarphi\|_2 %=\Omega(\sqrt{s_n/n})$.
%\qed
%\end{remark}

\begin{remark}
Dataset with multiple types
are now frequently collected for a common set of experimental subjects. This  new data structure is also called multimodal data. It is worth to mention, the above hypothesis test can be further extended to test the adequacy of multi-modal sparse linear regression model \citep{Li2021}. To be more specific, we consider the hypothesis test as follows:
\begin{align*}
        H_0 : \bY = \sum_{i=1}^{L}\bX_i \bbeta^{\star}_i + \sum_{i=L+1}^{M} \bX_i\bbeta^{\star}_i+ \cE \enspace \mathrm{versus} \enspace H_1 : \bY = \sum_{i=1}^{L}(\bF_i\bgamma_i^{\star}+\bU_i\bbeta_i^{\star})+ \sum_{i=L+1}^{M} \bX_i\bbeta_{i}^{\star}+ \cE.
\end{align*}
We aim at simultaneously testing whether the sparse regression is adequate for the any given $L$ modals. Here $\bX_i\in \mathbb{R}^{n\times d_i}$ is generated from the $i$-th modal, and possesses its own factor structure $\bX_i=\bF_i\bB^\top_i+\bU_i, i\in[M].$ Interested readers are referred to Appendix \textcolor{red}{D.4} for more details.\qed
\end{remark}

\section{Numerical Studies}
\label{simulation}
\subsection{Accuracy of Estimation}\label{simu_est}
%In this subsection, we present the experimental results for our theoretical findings given in section~\ref{est_factor}.
For data generation, we let number of factors $K=2$, dimension of covariate $d=1000$, $\bgamma^{\star}=(0.5,0.5)$, the first $s=3$ entries of $\bbeta^{\star}$ be $0.5$ and remaining $d-s$ entries be $0$. Throughout this subsection, we generate every entry of $\bF, \bU$ from the standard Gaussian distribution and let every entry of $\bB$ be generated from the uniform distribution Unif\,$(-1,1)$. We choose the noise distribution of $\varepsilon$ given in model~\eqref{eq_factor_linear} from (i) standard Gaussian, (ii) uniform, and (iii) $t_3$ distribution respectively.

\begin{figure}[ht]
    \centering
    \begin{tabular}{ccc}
   \hskip-30pt \includegraphics[width=0.35\textwidth]{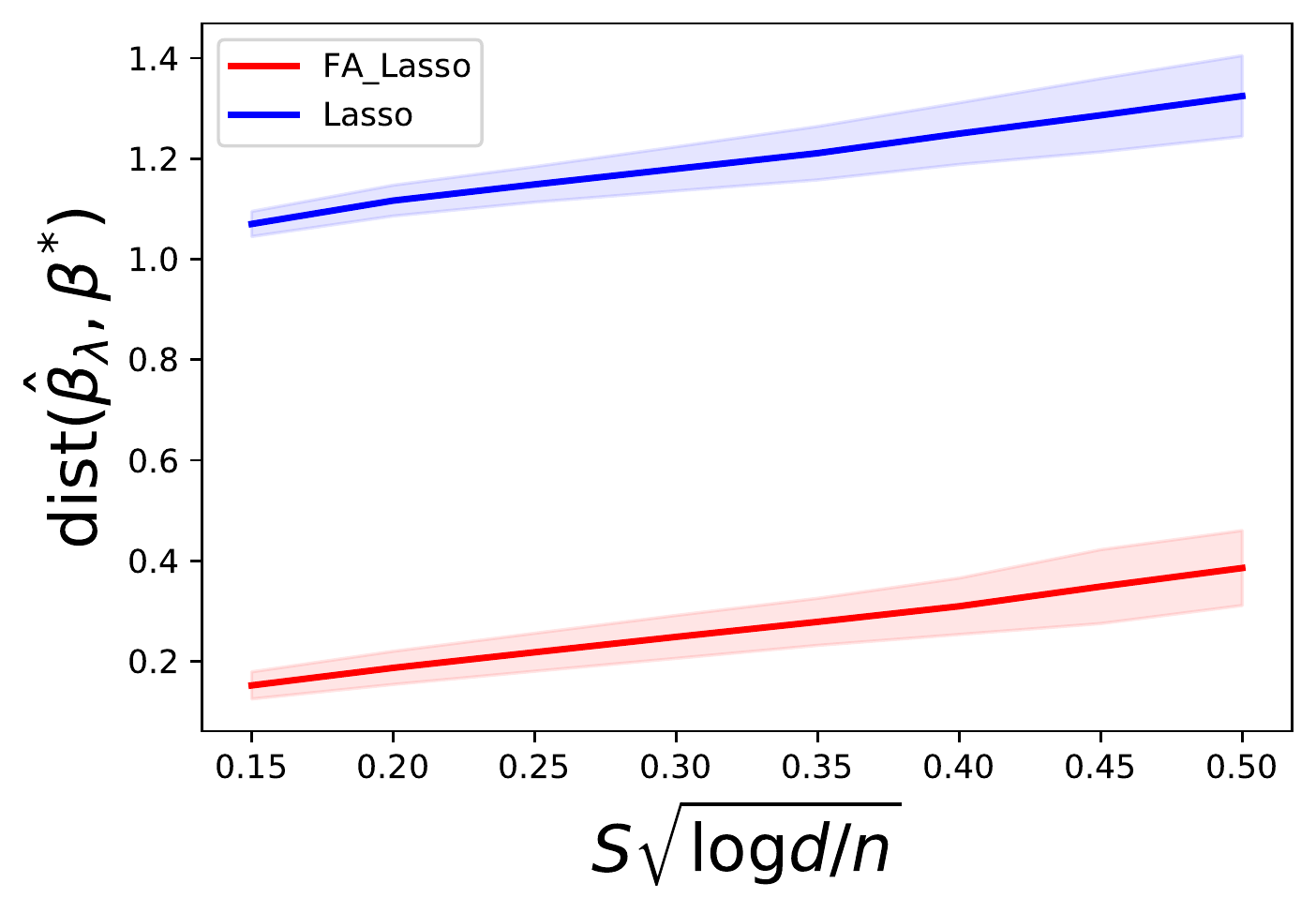}
   &\hskip-5pt \includegraphics[width=0.35\textwidth]{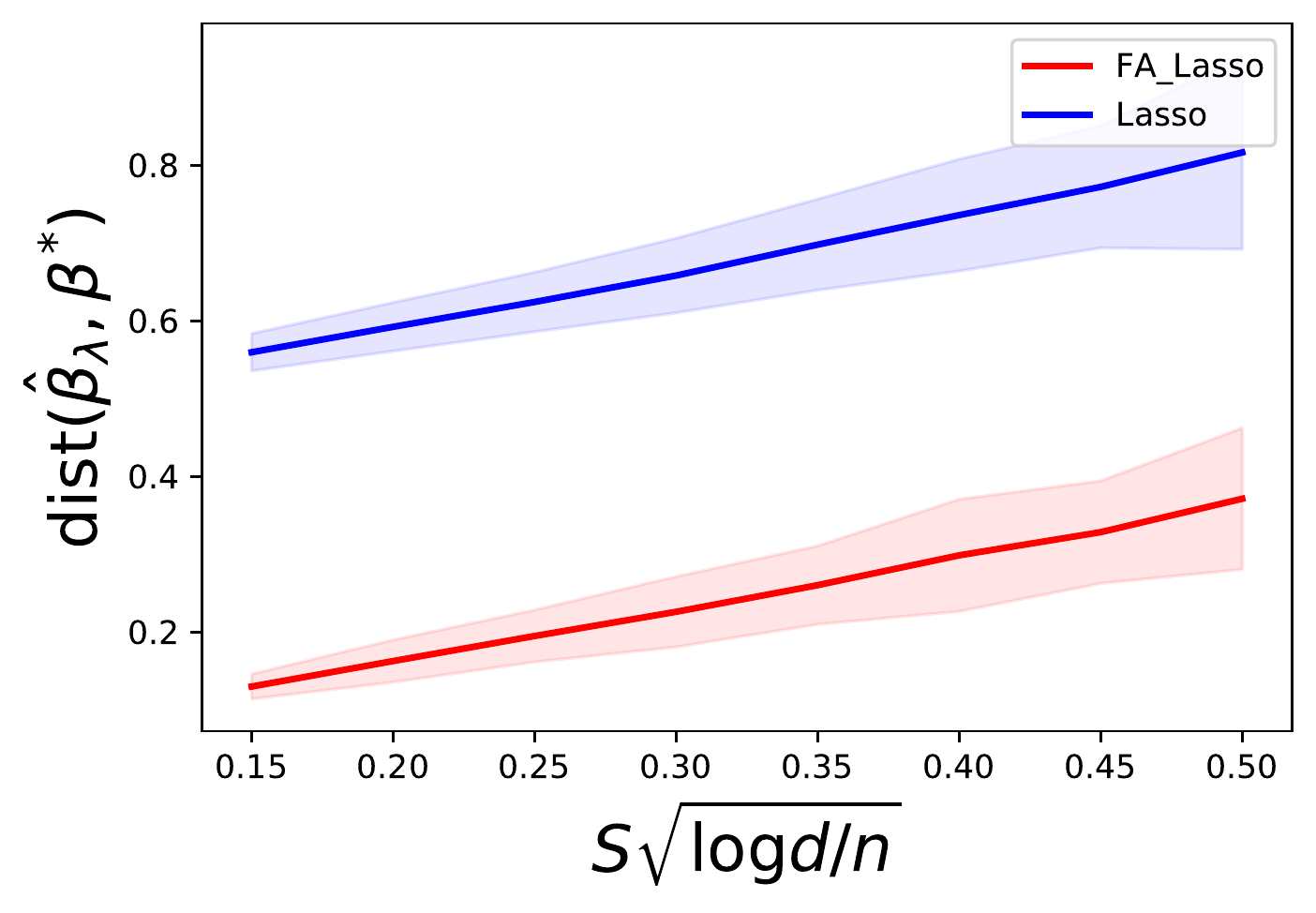}
   &\hskip-5pt \includegraphics[width=0.35\textwidth]{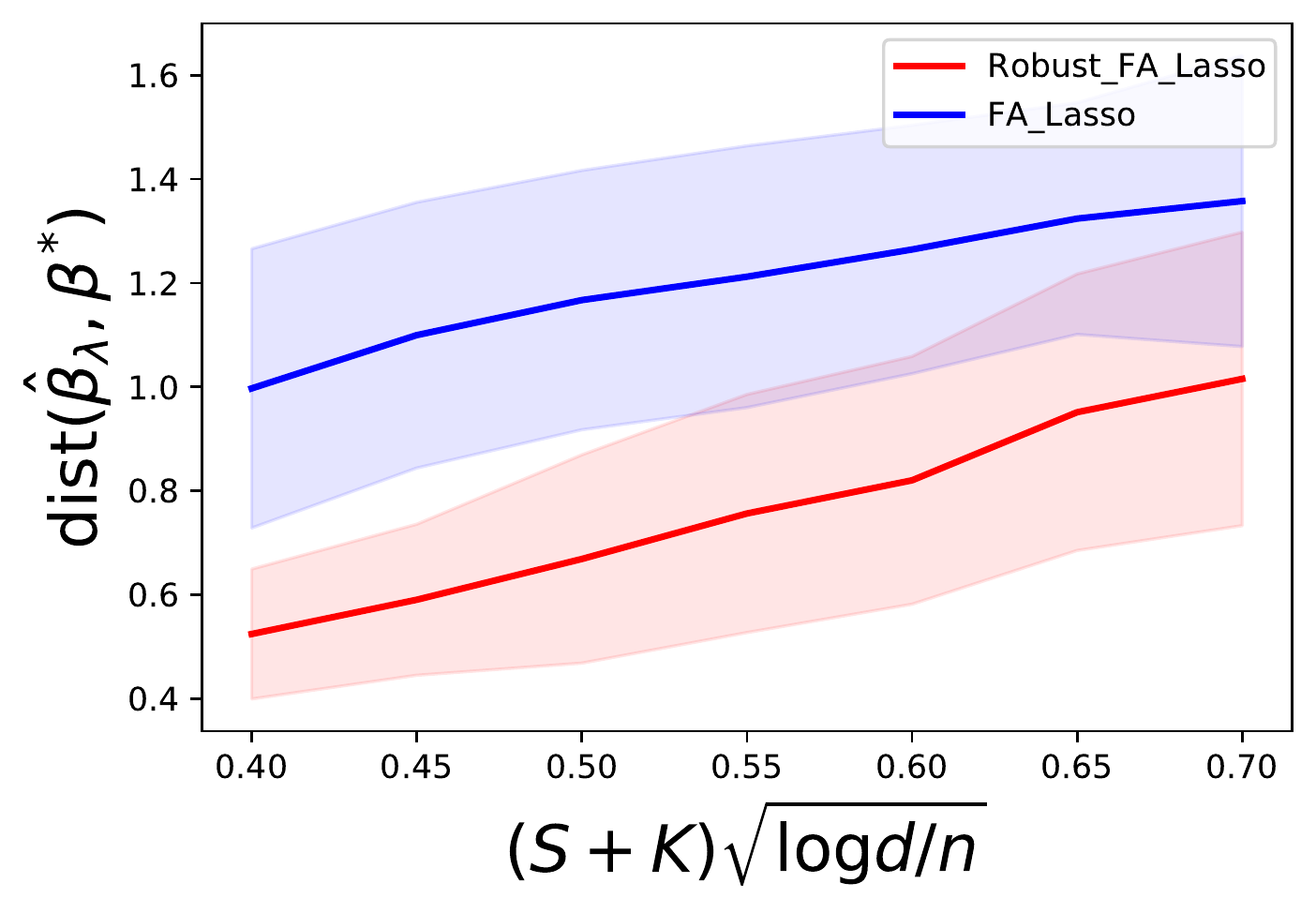}\\
   (a)&(b)&(c)
   \end{tabular}
    \caption{Accuracy for $\hat\beta_{\lambda}$ with dist$(\hat \beta_{\lambda},\beta^{\star}):=
    \|\hat\beta_{\lambda}-\beta^{\star}\|_1$ based on 500 replications. The light color regions indicate the standard errors across the simulation. Figure (a) and (b) depict the estimation results of model \eqref{factor_aug} with noise $\epsilon$ following the standard Gaussian and uniform distributions respectively. In (a) and (b), the red lines denote the estimation results using the method  \eqref{eq_Lasso} (labeled as FA\_Lasso in the figure) and the blue lines represent the results using Lasso with data $(\bX,\bY)$. In Figure (c),  the noise $\epsilon$ follows $t_{3}$-distribution. The red line in (c) represents the result of robust factor adjusted regression (Robust\_FA\_Lasso) via adaptive Huber estimation together with $\ell_1$-penalty given in \eqref{robust_regression} and the blue line represents the result achieved by using FA\_Lasso.}
    \label{fig:l1_light}
\end{figure}

Distributions (i) and (ii) have sub-Gaussian tails. For these two cases, we select sample size $n$ so that $s\sqrt{\log d/n}$ takes uniform grids in $[0.15, 0.5]$ . Then we generate $n$ response variables from model~\eqref{eq_factor_linear} and estimate our parameters via~\eqref{eq_Lasso}.
The results are shown as the red lines in Figure~\ref{fig:l1_light}. They lend further support to our theoretical findings given in section~\ref{est_factor} as the statistical rates there  are upper bounded by $O(s\sqrt{\log d/n})$. Moreover, we also show the estimation results by using Lasso directly on measurements $(\bX, \bY)$. Results are shown as the blue lines given in the first two figures in Figure \ref{fig:l1_light}.  Using Lasso directly on $(\bX,\bY)$ leads to much worse results due in part to the inadequacy of the model. In addition, as shown in \citet{Fan2020factor}, even when the sparse regression model is correct,  we still have better estimation accuracy using factor adjusted regression.

Distribution (iii) has only the bounded second moment. Likewise, we select corresponding number of observations $n$ so that $(s+K)\sqrt{\log d/n}$  takes uniform grids in $[0.4, 0.7]$. The reduced sample sizes help reduce the computation cost on the regularized adaptive Huber estimation using cross-validation to choose the parameter $\omega$.  We compare the results for the robust estimator~\eqref{robust_regression} with that of the factor adjusted regression~\eqref{eq_Lasso}.  The results are shown as the red and blue lines in part (c) of Figure~\ref{fig:l1_light} respectively. They provide stark evidence that  it is necessary to conduct the robust version of factor adjusted regression \eqref{robust_regression}  when noises have heavy tails.

\subsection{Adequacy of Factor Regression}\label{num_fact}

%In this subsection, we will illustrate the simulation results for section~\ref{pcr_test} and  further divide this subsection into data generation, algorithm implementation and result analysis respectively.

%\subsubsection{Data Generation Processes}
%\label{gen_sparse}
{\bf Data Generation Processes}.  We choose $n=200$,  $K=2$ and $d$ either 200 or 500 and the matrix $\bX=\bF\bB^\top+\bU$ using the following two models with entries of $\bB$ generated from Unif$(-1, 1)$.  %In each model, we generate $\bB$ in the same way as we did in section \ref{simu_est}.
\begin{enumerate}
	\item We generate every row of $\bF\in \mathbb{R}^{n\times K},\bU\in \mathbb{R}^{n\times d}$ from $\textrm{N}(\mathbf{0}, \bI_{K})$ and $\textrm{N}(\mathbf{0}, \bI_{d})$ respectively.
	\item  We let the $t$-th row $\bff_{t}\in \mathbb{R}^{K}$ of $\bF\in\mathbb{R}^{n\times K}$ follow $\bff_t=\bPhi\bff_{t-1}+\bxi_t$ where $\bPhi\in\mathbb{R}^{K\times K}$ with $\bPhi_{i,j}=0.5^{|i-j|+1},i,j\in[K]$. In addition, $\{\bxi_t\}_{t\ge 1}$ are drawn independently from $\textrm{N}(\textbf{0}, \bI_{K})$. We generate every row of $\bU$ from $N(\mathbf{0},\bSigma)$ where $\Sigma_{i,j}=0.6^{|i-j|},i,j\in[d]$.
	\end{enumerate}

 %��\textbf{(i).} We generate every row of $\bF\in \mathbb{R}^{n\times K},\bU\in \mathbb{R}^{n\times d}$ from $\textrm{N}(\mathbf{0},\mathbb{I}_{K\times K})$ and $\textrm{N}(\mathbf{0},\mathbb{I}_{d\times d})$ respectively. \\ \textbf{(ii).} We set every row of $\bF\in \mathbb{R}^{n\times K},\bU\in \mathbb{R}^{n\times d}$ from $\textrm{N}(\mathbf{0},\mathbb{I}_{K\times K})$ and $\textrm{N}(\mathbf{0},\Sigma_{d\times d})$  with $\Sigma_{i,j}=0.5^{|i-j|}$ respectively. \\ \textbf{(iii).} We let the $t$-th row $\bff_{t}\in \mathbb{R}^{K}$ of $\bF\in\mathbb{R}^{n\times K}$ follow $\bff_t=\bPhi\bff_{t-1}+\bxi_t$ where $\bPhi\in\mathbb{R}^{K\times K}$ with $\bPhi_{i,j}=0.5^{|i-j|+1},i,j\in[K]$. In addition, $\{\bxi_t\}_{t\ge 1}$ are drawn independently from $\textrm{N}(\textbf{0},\mathbb{I}_{K\times K})$. Moreover, we generate every row of $\bU$ from $N(0,\Sigma_{d\times d})$ where $\Sigma_{i,j}$ has been given above. \\

The response vector follows $\bY =\bF\bgamma^{\star}+ \bU \bbeta^{\star} + \cE$ in \eqref{Test_sparse_model} with  every entry of $\cE\in \mathbb{R}^{n}$ being generated independently from either from  $N(0,0.5^2)$ or uniform distribution Unif\,$(-\sqrt{3}/2,\sqrt{3}/2)$.
We set $\bgamma^{\star} = (0.5,0.5)$ and $\bbeta ^{\star}= (w,w,w,0,\cdots,0)$, where $w \geq 0$. When $w = 0$, the null hypothesis $\bY = \bF \bgamma^{\star} + \cE$ holds and the simulation results correpond to the size of the test.  Otherwise, they correspond to the power of the test.

\medskip
%\subsubsection{Implementation} \label{alg_sparse}

\noindent{\bf Implementation}.  We summarize the details of the proposed test
\begin{enumerate}\itemsep -0.03in
\item Estimate factor $\hat\bF$,  loading matrix $\hat\bB$, and noise $\hat \bU$ as in section \ref{est_factor}.

\item For the given $\hat\bU,$ we estimate $\hat\bTheta$ by using node-wise regression~\citep{Sara2014}.

\item%  Let $\tilde{\bY}=(\bI_{n}-\hat \bP)\bY$ with $\hat \bP=\hat \bF(\hat\bF^{\top} \hat\bF)^{-1}\hat \bF^{\top}$ and
Use Lasso method to estimate $\hat\bbeta_{\lambda}$ based on data $(\hat \bU, \tilde{\bY})$ with $\lambda$  chosen by  cross-validation.

\item Construct a de-biased $\tilde{\bbeta}_{\lambda}= \hbbeta_\lambda + n^{-1}\hat{\bTheta} \hat{\bU}^\top (\bY - \hat{\bU} \hbbeta_\lambda)$ .

\item  Obtain  $\hat{L}=1/\sqrt{n}\|\hat\bTheta\hat\bU^\top\tilde{\cE}\|_{\infty}$ with every entry of $\tilde{\cE}$ simulated from $N(0,1)$ and  repeat this computation $1000$ times to obtain the upper-$\alpha$ quantile $\hat c_{1-\alpha}$ of $\hat{L}$.

\item Compute the test result  $\xi_{\alpha}:=\mathbb{I}\{\sqrt{n}\|\tilde{\bbeta}_{\lambda}\|_{\infty}/\hat\sigma>\hat c_{1-\alpha}\}$ where $\hat\sigma$ is estimated by using refitted cross-validation~\citep{Fan2012variance}.
\end{enumerate}

%First, we divide $(\bY,\bX)$ evenly into two groups, i.e. $(\bY^{(1)},\bX^{(1)})$ and $(\bY^{(2)},\bX^{(2)})$. For each part of data, we estimate the factor $\hat \bF^{(i)}$ , factor loading $\hat\bB^{(i)}$ and noise $\hat\bU^{(i)}$ by following the method in section \ref{est_factor}, with $i\in\{1,2\}$. \\
%Second, for the first group of data $(\bY^{(1)},\bX^{(1)})$, we let $\tilde{\bY}^{(1)}=(\bI_{n}-\hat \bP^{(1)})\bY^{(1)}$, in which $\hat \bP^{(1)}=\hat \bF^{(1)}(\hat\bF^{(1)\top} \hat\bF^{(1)})^{-1}\hat \bF^{(1)\top}$ and use Iterative Sure Independence Screening(ISIS) method to select $\hat S_1$ via measurements $(\tilde{\bY}^{(1)},\hat \bU^{(1)}).$\\
%Third, for the second group of data, we construct $$  Q_n^{(2)} = \left\|\left(\bI_{n/2} - \bP_{\bX_{\hat{\cS}_1}}^{(2)}\right)\bY^{(2)}\right\|^2 - \left\|\left(\bI_{n/2} - \bP_{\hat{\bF}}^{(2)} -  \bP_{\hat{\bU}_{\hat{\cS}_1}}^{(2)}\right)\bY^{(2)}\right\|^2$$ in \eqref{Q_n2}. Here $\bP_{\bX_{\hat{\cS}_1}}^{(2)},\bP_{\hat{\bF}}^{(2)},\bP_{\hat{\bU}_{\hat{\cS}_1}}^{(2)} $ are the projection matrices constructed by $\bX_{\hat{\cS}_1},\hat{\bF}^{(2)}$ and $\hat{\bU}_{\hat{\cS}_1}$ respectively.\\
%Fourth, we estimate the variance of $\cE_i,i\ge 0$: $\hat\sigma^2$ by using refitted cross-validation.\\

The size and the power are then calculated based on 2000 simulations with  $\alpha = 0.05$ .

\medskip
%\subsubsection{Results}

\noindent{\bf Results}.
For $n=200, K=2$, $d\in \{200,500\}$ and all $w\in \{0,0.05,0.10,0.15,0.20\}$, we generate the data from each model and compute the testing results based on 2000 simulaitons with $\alpha = 0.05$.  The results are depicted in the Table~\ref{tab1}.  The column named Gaussian$(i)$, $i\in \{1,2\}$ represents the simulation results under model $i$ with Gaussian noise.  Similar labels applied to the uniform noise distribution.

\begin{table}[ht]
	\begin{center}
		\def\arraystretch{0.8}
		\setlength\tabcolsep{2pt}
		%\hspace*{-1.4cm}
		\begin{tabular}{c|c||c|c|c|c}
			%\hline
			& & Gaussian (1) & Gaussian (2)  &Uniform (1) &Uniform (2)\\
			%			\hhline{|=#=|=|=|}
			\hline
			\multirow{5}{*}{$p=200$	}  &$w= 0$ &$0.044$  & $0.047$ &$0.046$ &$0.048$ \\
			&$w=0.05$  &$0.067$ & $0.119$ &$0.065$ &$0.108$\\
			&$w=0.10$  &$0.326$ & $0.714$ &$0.311$ &$0.653$\\
			&$w=0.15$  &$0.859$ & $0.989$&$0.854$ &$0.984$\\
			&$w=0.20$  &$0.998$ & $1.000 $&$0.996$& $1.000$\\
			\hline
			\multirow{5}{*}{$p=500$	}  &$w= 0$ &$0.043$  &$0.040$ &$0.048$ & $0.436$\\
			&$w=0.05$  &$0.067$&$0.080$ &$0.059$&$0.071$\\
			&$w=0.10$  &$0.253$ &$0.632$ &$0.237$&$0.563$\\
			&$w=0.15$  &$0.787$ &$0.974$ &$0.780$&$0.962$\\
			&$w=0.20$  &$0.993$ &$1.000$ &$0.987$&$1.000$\\
		\hline
		\end{tabular}
	\end{center}
	\caption{Simulation results of section \ref{pcr_test} under different regimes.}
	\label{tab1}
\end{table}

Table~\ref{tab1} reveals that our test gives approximately the right size (subject to simulation error; see the rows with $w=0$). This is consistent with our theoretical findings given in section~\ref{pcr_test}. In addition, when $0<w<0.2$, the power of our test increases rapidly to $1$ which reveals the efficiency of our test statistic.

\subsection{Adequacy of Sparse Regression }\label{num_sparse}
This subsection  provides finite-sample validations for the results in section~\ref{sparse_test}.   We take the number of data used for screening $m=\lceil n^{0.8}\rceil$, use Iterative Sure Independence Screening method~\citep{Fan2008sure,saldana2018sis,Zhang19} to select $\hat{\cS}_1$ and apply the refitted cross-validation~\citep{Fan2012variance} to estimate $\sigma^{2}$. The size and the power of the test are computed based on 2000 simulations. 

\medskip

\noindent{\bf Data Generation Processes}.
We let  $n=250$,  $K=3$ and $d$ be either $250$ or $600$.
The noises $\varepsilon$ are i.i.d from   $N(0,0.5^2)$ or Unif\,$(-\sqrt{3}/2,\sqrt{3}/2)$.  The covariate  $\bX\in\mathbb{R}^{n\times d}$ follows the factor model $\bX=\bF\bB^\top+\bU$.  We generate $\bF,$ $\bU$ and $\bB$ in the same way as those in section~\ref{num_fact}.
In addition, the response variable follows $\bY =\bF\bvarphi^{\star}+ \bX \bbeta^{\star} + \cE$ in \eqref{Test_sparse_model} with  $\bbeta^{\star} = (0.8,0.8,0.8,0.8,0,\cdots,0)$ and $\bvarphi^{\star}=v\cdot \mathbf{1}_{K\times1}$ for several different values of $v \geq 0$. The case $v = 0$ corresponds to the null hypothesis and it is designed to test the validity of the size.

\medskip

\noindent{\bf Results}.
For $n=250$, $K=3$, $d\in \{250,600\}$ and $v\in \{0,0.04,0.08,0.12,0.16\}$, we implement the proposed method for every model in section~\ref{num_fact}. The simulation results are depicted in Table~\ref{tab2}. The column named Gaussian (or uniform) $(i)$, $i\in \{1,2\}$ represents the results  under  model $i$ with Gaussian (or uniform) noise mentioned in section \ref{num_fact}.  When $v=0$, the null hypothesis holds, our Type-I error is approximately $0.05$ which matches with the theoretical value.  In addition, when we increase the size of $v$ from $v=0.04$ to $v=0.16$, the power of our test statistic increases sharply to $1$, which reveals its efficiency.

 \begin{table}[ht]
	\begin{center}
		\def\arraystretch{0.8}
		\setlength\tabcolsep{4pt}
		%\hspace*{-1.4cm}
		\begin{tabular}{c|c||c|c|c|c}
			%\hline
			 & & Gaussian (1) & Gaussian (2)& Uniform (1) &Uniform (2)  \\
			%			\hhline{|=#=|=|=|}
			\hline
			\multirow{6}{*}{$p=250$	}  &$v= 0$ &$0.051$  & $0.054$&$0.056$ &$0.053$\\
			 			&$v=0.04$  &$0.215$ & $0.278$ &$0.233$&$0.286$\\
			 			 &$v=0.08$  &$0.659$ & $0.740$ &$0.655$&$0.750$\\
			 			 &$v=0.12$  &$0.965$ & $0.993$&$0.965$&$0.996$\\
			 			 &$v=0.16$  &$1.000$ & $1.000$&$1.000$&$1.000$\\
						 \hline
			\multirow{6}{*}{$p=600$	}  &$v= 0$ &$0.051$  &$0.052$&$0.050$ &$0.052$\\
&$v=0.04$  &$0.208$&$0.362$ &$0.197$&$0.353$\\
&$v=0.08$  &$0.624$ &$0.802$ &$0.604$&$0.785$\\
&$v=0.12$  &$0.941$ &$0.994$ &$0.934$&$0.999$\\
&$v=0.16$  &$1.000$ &$1.000$ &$0.999$&$1.000$\\ \hline
		\end{tabular}
	\end{center}
\caption{Simulation results of section \ref{sparse_test} under different regimes.}
\label{tab2}
\end{table}

We next discuss the necessity of using sample splitting.  Suppose we do not split samples and use the whole dataset to do sure screening and construct the test statistic. This will result in the high correlation between the selected set $\hat{\cS}$ and covariates when $\hat{\cS}$ is not a consistent estimator of $\cS_{\star}$. In this case, the asymptotic behavior of our test statistic is hard to capture. To demonstrate this point, we simulate the null distribution of the test statistic constructed without using sample splitting and compare it with the asymptotic distribution ($\chi_{K}^2$) via the quantile-quantile plot in  Figure~\ref{fig:qqplot}. Figure~\ref{fig:qqplot} reveals that the test statistic constructed without using sample splitting has heavier right tail than that of the $\chi_{K}^2$ distribution. The sizes of the test are much larger than the results in Table \ref{tab2} when $v=0$.

\begin{figure}[ht]
    \centering
    \includegraphics[width=0.35\textwidth]{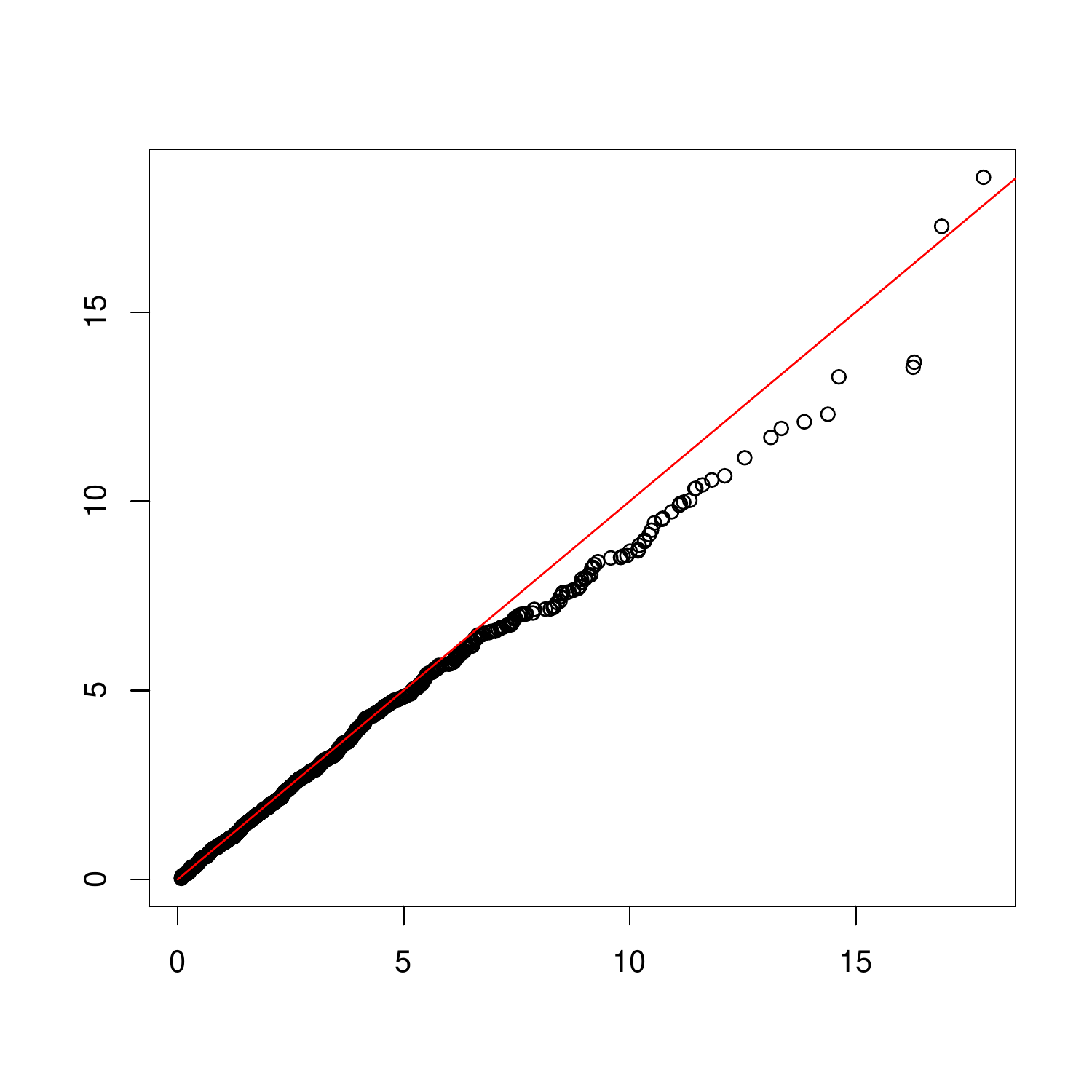}
    \caption{Quantiles of the $\chi_{K}^2$ distribution against those of the test statistic without sample splitting. The x-axis represents the quantiles of the test statistic whereas y-axis is the quantiles of $\chi_{K}^2$ distribution. }
    \label{fig:qqplot}
\end{figure}

\subsection{Empirical Applications}
\label{real_data}
In this section, we use a macroeconomic dataset named FRED-MD~\citep{McCracken2016} to illustrate the performance of our factor augmented regression  model (FARM) and investigate whether the latent factor regression model and sparse linear model are adequate.%For a detailed description about the FRED-MD dataset, we refer to~\cite{McCracken2016}.

There are 134 monthly U.S. macroeconomic variables in this dataset. As they measure certain aspects of economic health, these variables are driven by latent factors and hence correlated. They can be well explained by a few principal components.  In our study, we pick out two variables named 'HOUSTNE' and 'GS5' as our responses respectively and let the remaining variables be the covariates. Here 'HOUSTNE' represents the housing starts in the northeast region and 'GS5' denotes the $5$-year treasury rate.

There exist significant structural breaks for many variables around the year of financial crisis in 2008 which makes our data non-stationary even after performing the suggested transformations. Thus,  we analyze the dataset in two separate time periods independently. Specifically, we study the monthly data collected from February 1992 to October 2007 and from August 2010 to February 2020 respectively after examing the missingness and stationarity of the data.

To illustrate the performance of our proposed FARM against the sparse linear model and latent factor regression model, we first analyze the prediction results achieved by using these models. For every given time period and  model, we perform the prediction by using the moving window approach with window size 90 months. Indexing the panel data from 1 for each of the two time periods, for all $t>90$, we use the 90  previous measurements $\{(\bx_{t-90}, Y_{t-90}), \cdots, (\bx_{t-1}, Y_{t-1})\}$  to train a  model (FARM, sparse linear regression model, or latent factor regression model) and output a prediction $\hat{Y}_t$ as well as the in-sample average $\bar{Y}_t:=\frac{1}{90}\sum_{i=t-90}^{t-1}Y_i$. We measure the prediction accuracy by using out-of-sample $R^2$:
\begin{align*}
R^2=1-\frac{\sum_{t=91}^{T}(Y_t-\hat{Y}_t)^2}{\sum_{t=91}^{T}(Y_t-\bar{Y}_t)^2},
\end{align*}
%where $\bar{y}_t$ is the in sample average. \r{Q: in sample average up to $t$ or just in-sample average independent of $t$?}
where $T$ denotes the number of total data points in a given time period.
Table~\ref{tab3} presents the out-of-sample $R^2$ obtained by the aforementioned three models in the two time periods for predicting 'HOUSTNE' and 'GS5'.
Their detailed predictions are depicted in Figure~\ref{pred_house} and Figure~\ref{pred_sp500} respectively. Both out-of-sample $R^2$ and predictions depicted in Figures~\ref{pred_house} and~\ref{pred_sp500} show that FARM outperforms both latent factor regression and sparse linear regression models.

%For 'GS5', the corresponding values for our-of-sample $R^2$ are $(0.705,0.631,0.122)$ %$(0.5303731,0.2805115,-0.1045223)$
%and %$(0.7085637,0.7818674,-0.00114395)$
%$(0.657,0.540,0.233)$ respectively. %\{Q: somewhat hard to believe that sparse regression predict better}.
%Their detailed prediction performances are illustrated in Figure~\ref{pred_sp500}.
 \begin{table}[ht]
	\begin{center}
		\def\arraystretch{0.8}
		\setlength\tabcolsep{4pt}
		%\hspace*{-1.4cm}
		\begin{tabular}{c|c|c|c|c}
			%\hline
Time period	&Data	& FARM & SP\_Linear & LA\_Factor  \\
			%			\hhline{|=#=|=|
			\hline
			\hline
						\multirow{2}{*}{02.1992-10.2007}  &HOUSTNE  &$0.632$  & $0.584$ & $0.428$  \\
			&GS5  &$0.705$  &$0.631$ &$0.122$\\
			\hline
		%	\multirow{1}{*}{ HOUSTNE (02.1992-10.2007)	} &$0.632$  & $0.584$ & $0.428$\\
		%	\multirow{1}{*}{GS5 (02.1992-10.2007)	}  &$0.705$  &$0.631$ &$0.122$\\
								\multirow{2}{*}{08.2010-02.2020}  &HOUSTNE  &$0.694$  &$0.450$ &$0.219$  \\
		&GS5  &$0.657$  &$0.540$ &$0.233$\\
			%\multirow{1}{*}{HOUSTNE  (08.2010-02.2020)	}  &$0.694$  &$0.450$ &$0.219$\\
			%							\multirow{1}{*}{GS5  (08.2010-02.2020)	}  &$0.657$  &$0.540$ &$0.233$\\
						\hline
						
		\end{tabular}
	\end{center}
	\caption{Out-of-sample $R^2$ for predicting 'HOUSTNE' and 'GS5' data using different models in different time periods. %Each of the time period mentioned above contains the timestamps of both our training and testing data. 
LA\_Factor and SP\_Linear stand for latent factor regression and sparse linear regression respectively. } \label{tab3}
\end{table}

%for all variables% In addition, our FARM outperforms latent factor regression but behaves similar with sparse linear regression while predicting 'UNRATE' data.

We next conduct the hypothesis testing on the adequacy of latent factor regression and sparse linear regression respectively by using  FARM as the alternative model. As computing the bootstrap estimate of the null distribution is expensive for testing the adequacy of the factor model, we only conduct the hypothesis testing using the data in the entire two subperiods:  02.1992-10.2007 and 08.2010-02.2020.  The P-values for the tests are given in Table~\ref{tab4}. Taking the significant level $0.05$, the hypothesis testing results indicate that the latent factor regression is not adequate for all four different settings in Table~\ref{tab4}. As for sparse linear regression, it is accepted only for studying 'HOUSTNE' in the time period 02.1992-10.2007. These results match well with our prediction results.
%The implementation procedures are the same with those in section~\ref{num_fact} and section~\ref{num_sparse}.
\begin{table}[ht]
	\begin{center}
		\def\arraystretch{0.8}
		\setlength\tabcolsep{4pt}
		%\hspace*{-1.4cm}
		\begin{tabular}{c|c|c|c}
			%\hline
	Time period	&Data	& LA\_factor &SP\_Linear   \\
			%			\hhline{|=#=|=|
			\hline
			\hline
			\multirow{2}{*}{02.1992-10.2007}  &HOUSTNE   &$2.00\cdot 10^{-3}$  & $5.75\cdot 10^{-1}$ \\
			&GS5   &$ 2.50\cdot 10^{-3}$  &$8.73\cdot 10^{-3}$\\
			\hline
			\multirow{2}{*}{08.2010-02.2020}  &HOUSTNE   &$<10^{-3}$  &$2.00\cdot 10^{-2}$  \\
		&GS5  &$1.70\cdot 10^{-2}$  &$2.94\cdot 10^{-2}$\\
		\hline
		%	\multirow{1}{*}{LA\_Factor  (02.1992-10.2007)	} &$2.00\cdot 10^{-3}$  & $2.50\cdot 10^{-3}$\\
		%	\hline
		%	\multirow{1}{*}{SP\_Linear (02.1992-10.2007)	}  &$ 5.75\cdot 10^{-1}$  &$8.73\cdot 10^{-3}$\\
		%	\hline
		%	\multirow{1}{*}{LA\_Factor  (08.2010-02.2020)	}  &$<10^{-3}$  &$1.70\cdot 10^{-2}$\\
		%	\hline
		%								\multirow{1}{*}{SP\_Linear (08.2010-02.2020)	}  &$2.00\cdot 10^{-2}$  &$2.94\cdot 10^{-2}$\\
		%				\hline
							\end{tabular}
	\end{center}
	\caption{$p$-values for testing the adequacy of the latent factor regression and sparse linear regression models to explain 'HOUSTNE' and 'GS5' data in two different time periods.
The LA\_Factor and SP\_Linear have the same meaning as those in Table~\ref{tab3}.  }
	\label{tab4}
\end{table}

\iffalse
\begin{figure}[ht]
    \centering
    \includegraphics[width=0.9\textwidth]{draft_0612_2021/HOUSTNE_PREDICT.pdf}
    \caption{}
%%\label{fig:qqplot}
\end{figure}
\fi

\begin{figure}[]
	\centering
	\begin{tabular}{c}
		\hskip-5pt\includegraphics[width=0.45\textwidth]{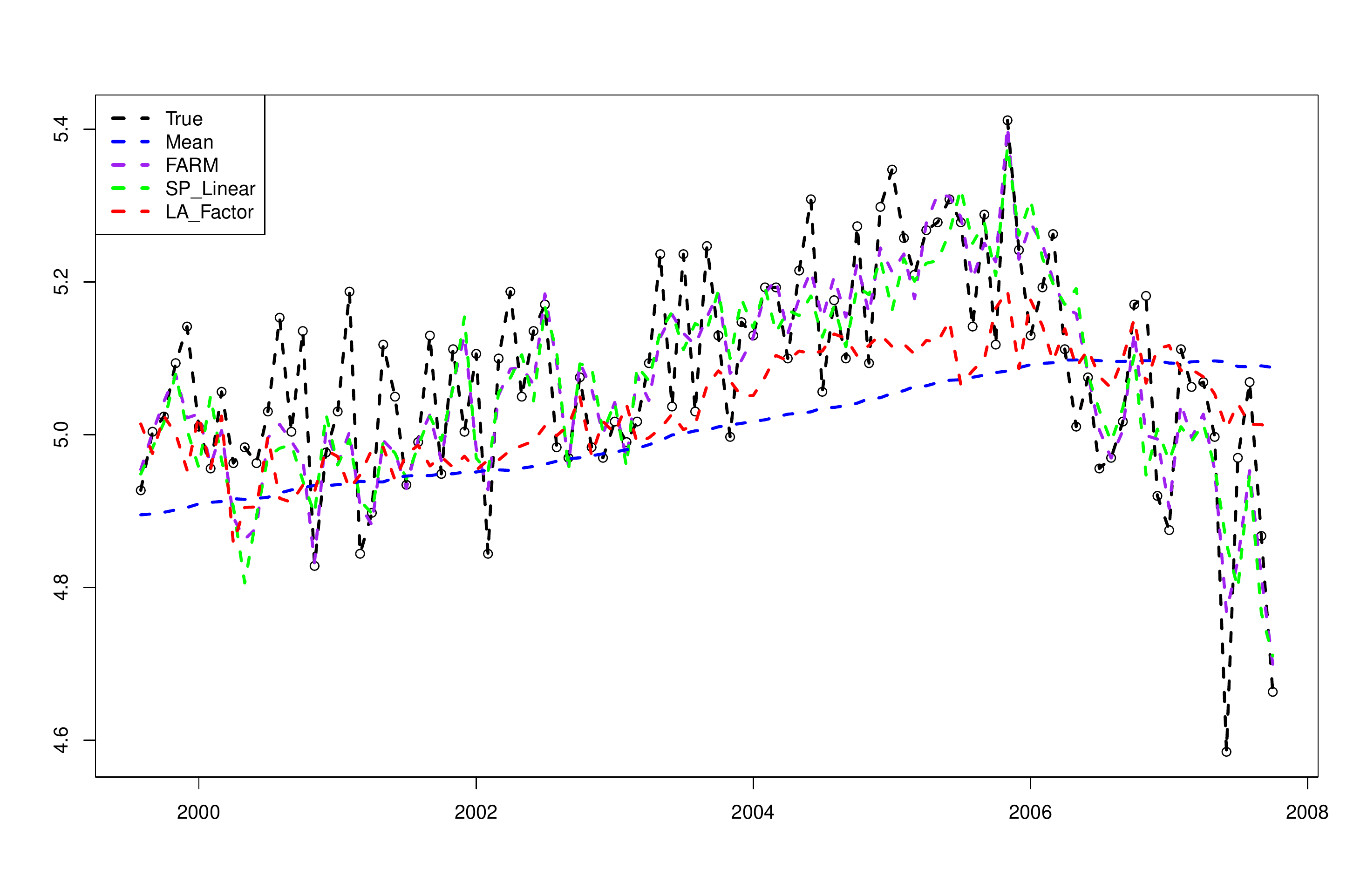}	 \hskip-5pt\includegraphics[width=0.45\textwidth]{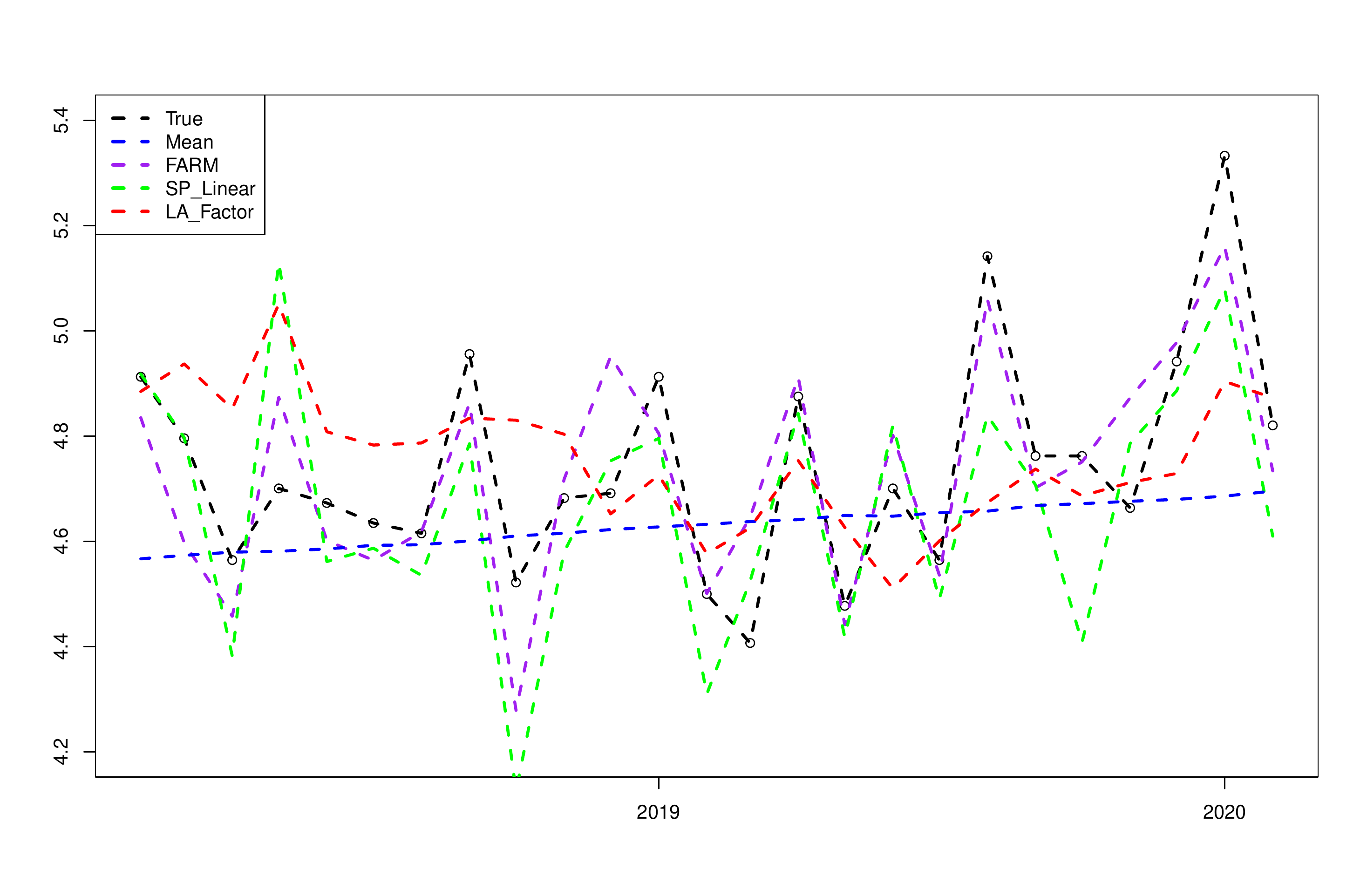}
		%		\hskip-30pt\includegraphics[width=0.30\textwidth]{sparse_null_k5.pdf}
		%		&
		%		\hskip-6pt\includegraphics[width=0.30\textwidth]{sparse_alterK5.pdf}\\
		%		(c)  & (d)\\
	\end{tabular}\\
	\caption{Out-of-sample prediction results for 'HOUSTNE' data  in time periods: August 1999 to October 2007 (left panel) and Februrary 2018 to February 2020 (right panel). %Here we draw these figures starting from the $91$-th timestamp in each of those original time periods.
The black dash line represents the true observed values, and the blue, purple, green and red dash lines represent the predictions made by using in-sample mean (moving average with  a window of 90), FARM, sparse linear model  and latent factor regression model, respectively. }
	\label{pred_house}
\end{figure}

\begin{figure}[]
	\centering
	\begin{tabular}{c} \hskip-5pt\includegraphics[width=0.45\textwidth]{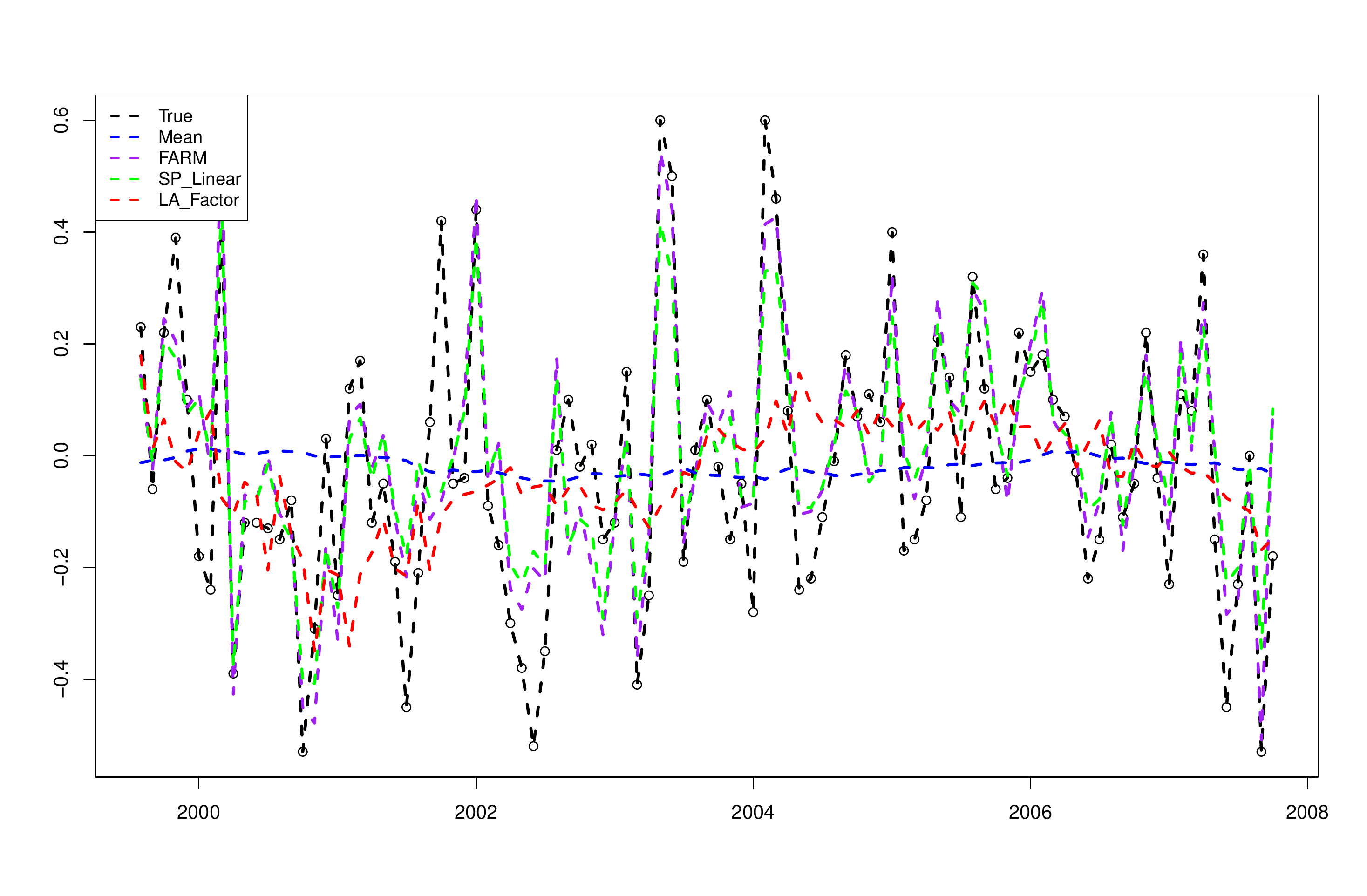}		\hskip-5pt\includegraphics[width=0.45\textwidth]{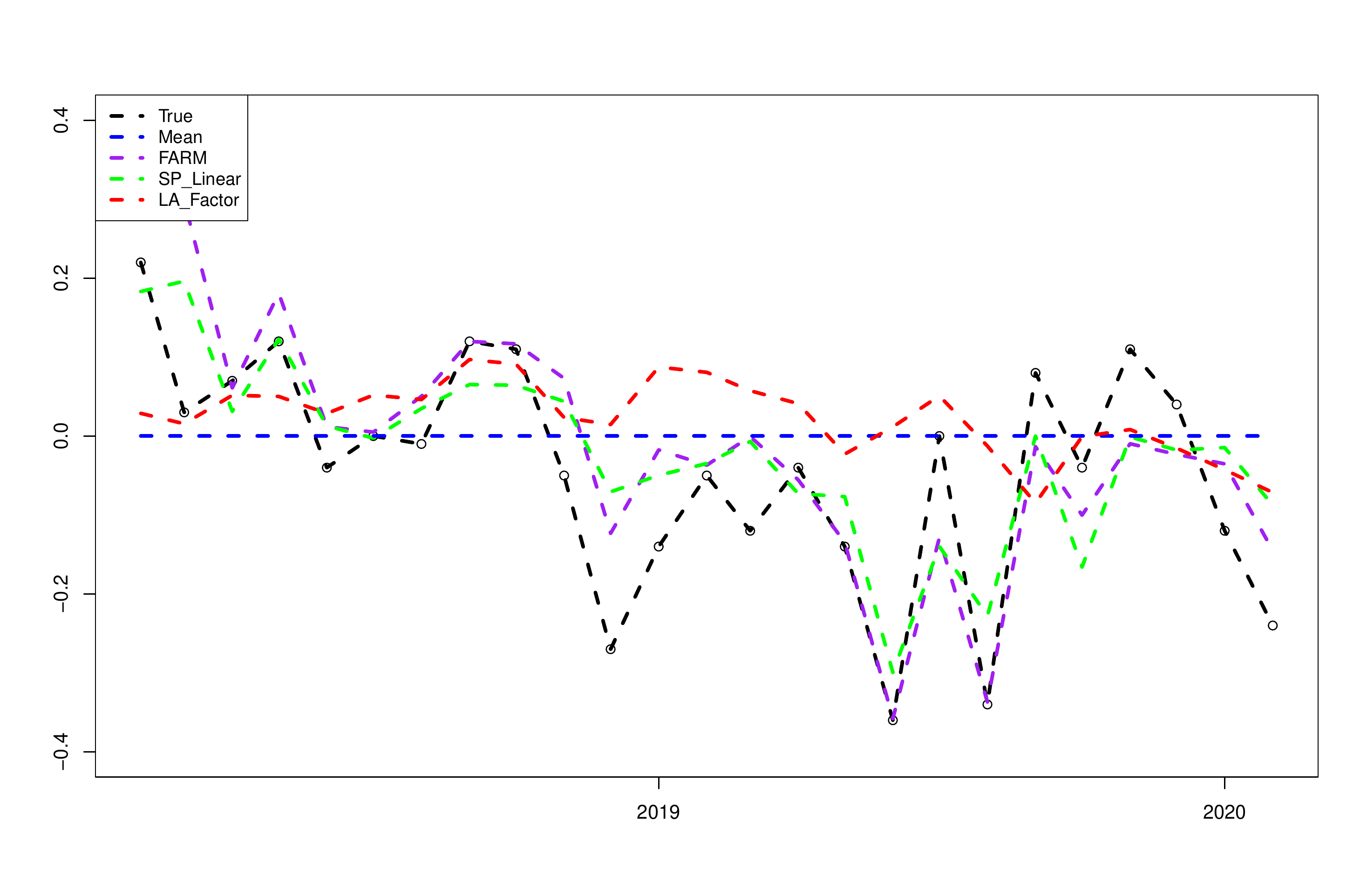}
		%		\hskip-30pt\includegraphics[width=0.30\textwidth]{sparse_null_k5.pdf}
		%		&
		%		\hskip-6pt\includegraphics[width=0.30\textwidth]{sparse_alterK5.pdf}\\
		%		(c)  & (d)\\
	\end{tabular}\\
	\caption{Out-of-sample prediction results for 'GS5' data  in time periods: August 1999 to October 2007 (left panel) and Februrary 2018 to February 2020 (right panel). The same captions as  those in Figure \ref{pred_house} are used. }
	\label{pred_sp500}
\end{figure}

\iffalse
\begin{figure}
    \centering
    \includegraphics[width=0.9\textwidth]{draft_0612_2021/UNRATE_PREDICT.pdf}
    \caption{}
%%\label{fig:qqplot}
\end{figure}
\fi

\section{Conclusion}
In this paper, we propose a model named Factor Augmented sparse linear Regression Model (FARM), which contains the latent factor regression and the sparse linear regression as our special cases. The model expands the space spanned by covariates into useful directions and hence use additional information beyond the linear space spanned by the predictors. We provide theoretical guarantees for our model estimation under the existence of light-tailed and heavy-tailed noises respectively. In addition, we leverage our model as the alternative one to test the sufficiency of the latent factor regression model and sparse regression model. We believe that the study is among the first of this kind.  The practical performance of our model estimation and our constructed test statistics are proven by extensive simulation studies including both synthetic data and real data. Moreover, it is worth to mention that our model and methodology can be  extended to more general supervised learning problems such as nonparametric regression, quantile regression, regression and classification trees, support vector machines, among others where the factor augmentation idea is always useful. 

\begin{singlespace}
\small
\bibliographystyle{abbrvnat}

\bibliography{reference}

\begin{thebibliography}{70}
\providecommand{\natexlab}[1]{#1}
\providecommand{\url}[1]{\texttt{#1}}
\expandafter\ifx\csname urlstyle\endcsname\relax
  \providecommand{\doi}[1]{doi: #1}\else
  \providecommand{\doi}{doi: \begingroup \urlstyle{rm}\Url}\fi

\bibitem[Ahn and Horenstein(2013)]{Ahn2013}
S.~C. Ahn and A.~R. Horenstein.
\newblock Eigenvalue ratio test for the number of factors.
\newblock \emph{Econometrica}, 81\penalty0 (3):\penalty0 1203--1227, 2013.

\bibitem[Avella-Medina et~al.(2018)Avella-Medina, Battey, Fan, and
  Li]{avella2018robust}
M.~Avella-Medina, H.~S. Battey, J.~Fan, and Q.~Li.
\newblock Robust estimation of high-dimensional covariance and precision
  matrices.
\newblock \emph{Biometrika}, 105\penalty0 (2):\penalty0 271--284, 2018.

\bibitem[Bai(2003)]{Bai2003inferential}
J.~Bai.
\newblock Inferential theory for factor models of large dimensions.
\newblock \emph{Econometrica}, 71\penalty0 (1):\penalty0 135--171, 2003.

\bibitem[Bai and Li(2012)]{Bai2012statistical}
J.~Bai and K.~Li.
\newblock Statistical analysis of factor models of high dimension.
\newblock \emph{Ann. Statist.}, 40\penalty0 (1):\penalty0 436--465, 2012.

\bibitem[Bai and Ng(2002)]{Bai2002}
J.~Bai and S.~Ng.
\newblock Determining the number of factors in approximate factor models.
\newblock \emph{Econometrica}, 70\penalty0 (1):\penalty0 191--221, 2002.

\bibitem[Bai and Ng(2006)]{Bai2006}
J.~Bai and S.~Ng.
\newblock Confidence intervals for diffusion index forecasts and inference for
  factor-augmented regressions.
\newblock \emph{Econometrica}, 74\penalty0 (4):\penalty0 1133--1150, 2006.

\bibitem[Bai and Ng(2008)]{Bai2008}
J.~Bai and S.~Ng.
\newblock Forecasting economic time series using targeted predictors.
\newblock \emph{Journal of Econometrics}, 146\penalty0 (2):\penalty0 304--317,
  2008.
\newblock ISSN 0304-4076.

\bibitem[Bair et~al.(2006)Bair, Hastie, Paul, and Tibshirani]{Bair_Hastie_2006}
E.~Bair, T.~Hastie, D.~Paul, and R.~Tibshirani.
\newblock Prediction by supervised principal components.
\newblock \emph{Journal of the American Statistical Association}, 101\penalty0
  (473):\penalty0 119--137, 2006.

\bibitem[Barut et~al.(2016)Barut, Fan, and Verhasselt]{Barut2016conditional}
E.~Barut, J.~Fan, and A.~Verhasselt.
\newblock Conditional sure independence screening.
\newblock \emph{J. Amer. Statist. Assoc.}, 111\penalty0 (515):\penalty0
  1266--1277, 2016.

\bibitem[Belloni and Chernozhukov(2011)]{Belloni2011}
A.~Belloni and V.~Chernozhukov.
\newblock $\ell_{1}$-penalized quantile regression in high-dimensional sparse
  models.
\newblock \emph{Ann. Statist.}, 39\penalty0 (1):\penalty0 82--130, 2011.

\bibitem[Bing et~al.(2019)Bing, Bunea, and Wegkamp]{Benea2020_2}
X.~Bing, F.~Bunea, and M.~Wegkamp.
\newblock Inference in latent factor regression with clusterable features.
\newblock \emph{arXiv:1905.12696}, 2019.

\bibitem[Bing et~al.(2021)Bing, Bunea, Strimas-Mackey, and Wegkamp]{Bing2021}
X.~Bing, F.~Bunea, S.~Strimas-Mackey, and M.~Wegkamp.
\newblock Prediction under latent factor regression: Adaptive pcr,
  interpolating predictors and beyond.
\newblock \emph{Journal of Machine Learning Research}, 22\penalty0
  (177):\penalty0 1--50, 2021.

\bibitem[Bunea et~al.(2020)Bunea, Strimas-Mackey, and Wegkamp]{Benea2020}
F.~Bunea, S.~Strimas-Mackey, and M.~Wegkamp.
\newblock Interpolating predictors in high-dimensional factor regression.
\newblock \emph{arXiv:2002.02525}, 2020.

\bibitem[Cai et~al.(2011)Cai, Liu, and Luo]{Cai2011}
T.~Cai, W.~Liu, and X.~Luo.
\newblock A constrained $\ell_1$ minimization approach to sparse precision
  matrix estimation.
\newblock \emph{J. Amer. Statist. Assoc.}, 106\penalty0 (494):\penalty0
  594--607, 2011.

\bibitem[Candes and Tao(2007)]{CandesTao2007}
E.~Candes and T.~Tao.
\newblock The {D}antzig selector: statistical estimation when {$p$} is much
  larger than {$n$}.
\newblock \emph{Ann. Statist.}, 35\penalty0 (6):\penalty0 2313--2351, 2007.

\bibitem[Chernozhukov et~al.(2013)Chernozhukov, Chetverikov, and Kato]{CCK2013}
V.~Chernozhukov, D.~Chetverikov, and K.~Kato.
\newblock Gaussian approximations and multiplier bootstrap for maxima of sums
  of high-dimensional random vectors.
\newblock \emph{Ann. Statist.}, 41\penalty0 (6):\penalty0 2786--2819, 2013.

\bibitem[Chernozhukov et~al.(2017)Chernozhukov, Chetverikov, and Kato]{CCK2017}
V.~Chernozhukov, D.~Chetverikov, and K.~Kato.
\newblock Central limit theorems and bootstrap in high dimensions.
\newblock \emph{Ann. Probab.}, 45\penalty0 (4):\penalty0 2309--2352, 2017.

\bibitem[Chernozhukov et~al.(2020)Chernozhukov, Chetverikov, and
  Koike]{CCK2020}
V.~Chernozhukov, D.~Chetverikov, and Y.~Koike.
\newblock Nearly optimal central limit theorem and bootstrap approximations in
  high dimensions.
\newblock \emph{arXiv preprint arXiv:2012.09513}, 2020.

\bibitem[Chu et~al.(2016)Chu, Li, and Reimherr]{Chu2016feature}
W.~Chu, R.~Li, and M.~Reimherr.
\newblock Feature screening for time-varying coefficient models with ultrahigh
  dimensional longitudinal data.
\newblock \emph{Ann. Appl. Stat.}, 10\penalty0 (2):\penalty0 596, 2016.

\bibitem[Dezeure et~al.(2017)Dezeure, B{\"u}hlmann, and Zhang]{Dezeure2017}
R.~Dezeure, P.~B{\"u}hlmann, and C.-H. Zhang.
\newblock High-dimensional simultaneous inference with the bootstrap.
\newblock \emph{Test}, 26\penalty0 (4):\penalty0 685--719, 2017.

\bibitem[Efron et~al.(2004)Efron, Hastie, Johnstone, and Tibshirani]{Efron2004}
B.~Efron, T.~Hastie, I.~Johnstone, and R.~Tibshirani.
\newblock Least angle regression.
\newblock \emph{Ann. Statist.}, 32\penalty0 (2):\penalty0 407--499, 2004.

\bibitem[Fan and Li(2001)]{SCAD2001}
J.~Fan and R.~Li.
\newblock Variable selection via nonconcave penalized likelihood and its oracle
  properties.
\newblock \emph{J. Amer. Statist. Assoc.}, 96\penalty0 (456):\penalty0
  1348--1360, 2001.

\bibitem[Fan and Liao(2020)]{Liao2020}
J.~Fan and Y.~Liao.
\newblock Learning latent factors from diversified projections and its
  applications to over-estimated and weak factors.
\newblock \emph{Journal of the American Statistical Association}, 0\penalty0
  (0):\penalty0 1--16, 2020.

\bibitem[Fan and Lv(2008)]{Fan2008sure}
J.~Fan and J.~Lv.
\newblock Sure independence screening for ultrahigh dimensional feature space.
\newblock \emph{J. R. Stat. Soc. Ser. B Stat. Methodol.}, 70\penalty0
  (5):\penalty0 849--911, 2008.

\bibitem[Fan and Lv(2011)]{FanLv2011}
J.~Fan and J.~Lv.
\newblock Nonconcave penalized likelihood with {NP}-dimensionality.
\newblock \emph{IEEE Trans. Inform. Theory}, 57\penalty0 (8):\penalty0
  5467--5484, 2011.

\bibitem[Fan and Song(2010)]{FanSong2010}
J.~Fan and R.~Song.
\newblock Sure independence screening in generalized linear models with
  {NP}-dimensionality.
\newblock \emph{Ann. Statist.}, 38\penalty0 (6):\penalty0 3567--3604, 2010.

\bibitem[Fan et~al.(2012)Fan, Guo, and Hao]{Fan2012variance}
J.~Fan, S.~Guo, and N.~Hao.
\newblock Variance estimation using refitted cross-validation in ultrahigh
  dimensional regression.
\newblock \emph{J. R. Stat. Soc. Ser. B. Stat. Methodol.}, 74\penalty0
  (1):\penalty0 37--65, 2012.

\bibitem[Fan et~al.(2013)Fan, Liao, and Mincheva]{Fan2013POET}
J.~Fan, Y.~Liao, and M.~Mincheva.
\newblock Large covariance estimation by thresholding principal orthogonal
  complements.
\newblock \emph{J. R. Stat. Soc. Ser. B. Stat. Methodol.}, 75\penalty0
  (4):\penalty0 603--680, 2013.
\newblock With 33 discussions by 57 authors and a reply by Fan, Liao and
  Mincheva.

\bibitem[Fan et~al.(2014)Fan, Fan, and Barut]{Fan2014adaptive}
J.~Fan, Y.~Fan, and E.~Barut.
\newblock Adaptive robust variable selection.
\newblock \emph{Ann. Statist.}, 42\penalty0 (1):\penalty0 324, 2014.

\bibitem[Fan et~al.(2017{\natexlab{a}})Fan, Li, and Wang]{Fan2017estimation}
J.~Fan, Q.~Li, and Y.~Wang.
\newblock Estimation of high dimensional mean regression in the absence of
  symmetry and light tail assumptions.
\newblock \emph{J. R. Stat. Soc. Ser. B Stat. Methodol.}, 79\penalty0
  (1):\penalty0 247, 2017{\natexlab{a}}.

\bibitem[Fan et~al.(2017{\natexlab{b}})Fan, Xue, and Yao]{FAN_Xue_Yao}
J.~Fan, L.~Xue, and j.~Yao.
\newblock Sufficient forecasting using factor models.
\newblock \emph{Journal of Econometrics}, 201\penalty0 (2):\penalty0 292--306,
  2017{\natexlab{b}}.
\newblock ISSN 0304-4076.

\bibitem[Fan et~al.(2020{\natexlab{a}})Fan, Ke, and Wang]{Fan2020factor}
J.~Fan, Y.~Ke, and K.~Wang.
\newblock Factor-adjusted regularized model selection.
\newblock \emph{J. Econometrics}, 216\penalty0 (1):\penalty0 71--85,
  2020{\natexlab{a}}.

\bibitem[Fan et~al.(2020{\natexlab{b}})Fan, Li, Zhang, and
  Zou]{Fan_Li_Zhang_Zou_data}
J.~Fan, R.~Li, C.-H. Zhang, and H.~Zou.
\newblock Statistical foundations of data science.
\newblock 2020{\natexlab{b}}.

\bibitem[Fan et~al.(2021{\natexlab{a}})Fan, Guo, and Zheng]{fan2021estimating}
J.~Fan, J.~Guo, and S.~Zheng.
\newblock Estimating number of factors by adjusted eigenvalues thresholding.
\newblock \emph{Journal of the American Statistical Association}, pages 1--10,
  2021{\natexlab{a}}.

\bibitem[Fan et~al.(2021{\natexlab{b}})Fan, Masini, and
  Medeiros]{Fan2021_factor}
J.~Fan, R.~Masini, and M.~C. Medeiros.
\newblock Bridging factor and sparse models.
\newblock \emph{arXiv:2102.11341}, 2021{\natexlab{b}}.

\bibitem[Fan et~al.(2021{\natexlab{c}})Fan, Yang, and Yu]{Yu2021}
J.~Fan, Z.~Yang, and M.~Yu.
\newblock Understanding implicit regularization in over-parameterized single
  index model.
\newblock \emph{arXiv:2007.08322v3}, 2021{\natexlab{c}}.

\bibitem[Hernan and Robins(2019)]{hernan2019causal}
M.~Hernan and J.~Robins.
\newblock \emph{Causal Inference}.
\newblock Chapman \& Hall/CRC Monographs on Statistics \& Applied Probab. CRC
  Press LLC, 2019.
\newblock ISBN 9781420076165.

\bibitem[Hotelling(1933)]{Hotelling1933analysis}
H.~Hotelling.
\newblock Analysis of a complex of statistical variables into principal
  components.
\newblock \emph{Journal of educational psychology}, 24\penalty0 (6):\penalty0
  417, 1933.

\bibitem[Imbens and Rubin(2015)]{Imbens2015}
G.~Imbens and D.~Rubin.
\newblock Causal inference for statistics, social, and biomedical sciences: An
  introduction.
\newblock 2015.

\bibitem[Javanmard and Montanari(2014)]{Javanmard2014}
A.~Javanmard and A.~Montanari.
\newblock Confidence intervals and hypothesis testing for high-dimensional
  regression.
\newblock \emph{J. Mach. Learn. Res.}, 15:\penalty0 2869--2909, 2014.

\bibitem[Jolliffe(1982)]{Jolliffe1982}
I.~T. Jolliffe.
\newblock A note on the use of principal components in regression.
\newblock \emph{Journal of the Royal Statistical Society. Series C (Applied
  Statistics)}, 31\penalty0 (3):\penalty0 300--303, 1982.

\bibitem[Lam and Yao(2012)]{Lam2012}
C.~Lam and Q.~Yao.
\newblock {Factor modeling for high-dimensional time series: Inference for the
  number of factors}.
\newblock \emph{The Annals of Statistics}, 40\penalty0 (2):\penalty0 694 --
  726, 2012.

\bibitem[Li et~al.(2012)Li, Peng, Zhang, and Zhu]{Li2012robust}
G.~Li, H.~Peng, J.~Zhang, and L.~Zhu.
\newblock Robust rank correlation based screening.
\newblock \emph{Ann. Statist.}, 40\penalty0 (3):\penalty0 1846--1877, 2012.

\bibitem[Li and Li(2021)]{Li2021}
Q.~Li and L.~Li.
\newblock Integrative factor regression and its inference for multimodal data
  analysis.
\newblock \emph{J. Amer. Statist. Assoc.}, 113\penalty0 (521):\penalty0 1--15,
  2021.

\bibitem[Li et~al.(2018)Li, Cheng, Fan, and Wang]{Li2018embracing}
Q.~Li, G.~Cheng, J.~Fan, and Y.~Wang.
\newblock Embracing the blessing of dimensionality in factor models.
\newblock \emph{J. Amer. Statist. Assoc.}, 113\penalty0 (521):\penalty0
  380--389, 2018.

\bibitem[Liu et~al.(2014)Liu, Li, and Wu]{Liu2014feature}
J.~Liu, R.~Li, and R.~Wu.
\newblock Feature selection for varying coefficient models with
  ultrahigh-dimensional covariates.
\newblock \emph{J. Amer. Statist. Assoc.}, 109\penalty0 (505):\penalty0
  266--274, 2014.

\bibitem[Loh and Wainwright(2012)]{Loh2012}
P.-L. Loh and M.~J. Wainwright.
\newblock High-dimensional regression with noisy and missing data: Provable
  guarantees with nonconvexity.
\newblock \emph{Ann. Statist.}, 40\penalty0 (3):\penalty0 1637--1664, 2012.

\bibitem[McCracken and Ng(2016)]{McCracken2016}
M.~W. McCracken and S.~Ng.
\newblock Fred-md: A monthly database for macroeconomic research.
\newblock \emph{Journal of Business \& Economic Statistics}, 34\penalty0
  (4):\penalty0 574--589, 2016.

\bibitem[Nickl and Van De~Geer(2013)]{Nickl2013confidence}
R.~Nickl and S.~Van De~Geer.
\newblock Confidence sets in sparse regression.
\newblock \emph{Ann. Statist.}, 41\penalty0 (6):\penalty0 2852--2876, 2013.

\bibitem[Peng et~al.(2016)Peng, Wang, and Wu]{Peng2016}
B.~Peng, L.~Wang, and Y.~Wu.
\newblock An error bound for ${L}_{1}$-norm support vector machine coefficients
  in ultra-high dimension.
\newblock \emph{J. Mach. Learn. Res.}, 17\penalty0 (1):\penalty0 8279--8304,
  2016.

\bibitem[Saldana and Feng(2018)]{saldana2018sis}
D.~F. Saldana and Y.~Feng.
\newblock Sis: An r package for sure independence screening in
  ultrahigh-dimensional statistical models.
\newblock \emph{Journal of Statistical Software}, 83\penalty0 (2):\penalty0
  1--25, 2018.

\bibitem[Shi et~al.(2019)Shi, Song, Chen, and Li]{Shi2019}
C.~Shi, R.~Song, Z.~Chen, and R.~Li.
\newblock Linear hypothesis testing for high dimensional generalized linear
  models.
\newblock \emph{Ann. Statist.}, 47\penalty0 (5):\penalty0 2671--2703, 2019.

\bibitem[Stock and Watson(2002)]{Stock_Watson}
J.~H. Stock and M.~W. Watson.
\newblock Forecasting using principal components from a large number of
  predictors.
\newblock \emph{Journal of the American Statistical Association}, 97\penalty0
  (460):\penalty0 1167--1179, 2002.

\bibitem[Sun et~al.(2020)Sun, Zhou, and Fan]{AHR2020}
Q.~Sun, W.-X. Zhou, and J.~Fan.
\newblock Adaptive {H}uber regression.
\newblock \emph{Journal of the American Statistical Association}, 115\penalty0
  (529):\penalty0 254--265, 2020.

\bibitem[Sun and Zhang(2012)]{Sun2012scaled}
T.~Sun and C.-H. Zhang.
\newblock Scaled sparse linear regression.
\newblock \emph{Biometrika}, 99\penalty0 (4):\penalty0 879--898, 2012.

\bibitem[Tibshirani(1996)]{LASSO1996}
R.~Tibshirani.
\newblock Regression shrinkage and selection via the lasso.
\newblock \emph{J. Roy. Statist. Soc. Ser. B}, 58\penalty0 (1):\penalty0
  267--288, 1996.

\bibitem[Van~de Geer(2008)]{Van2008high}
S.~Van~de Geer.
\newblock High-dimensional generalized linear models and the lasso.
\newblock \emph{Ann. Statist.}, 36\penalty0 (2):\penalty0 614--645, 2008.

\bibitem[van~de Geer et~al.(2014)van~de Geer, B{\"u}hlmann, Ritov, and
  Dezeure]{Sara2014}
S.~van~de Geer, P.~B{\"u}hlmann, Y.~Ritov, and R.~Dezeure.
\newblock On asymptotically optimal confidence regions and tests for
  high-dimensional models.
\newblock \emph{Ann. Statist.}, 42\penalty0 (3):\penalty0 1166--1202, 2014.

\bibitem[Wang and Fan(2017)]{Wang2017}
W.~Wang and J.~Fan.
\newblock Asymptotics of empirical eigenstructure for high dimensional spiked
  covariance.
\newblock \emph{Ann. Statist.}, 45\penalty0 (3):\penalty0 1342--1374, 2017.

\bibitem[Wang and Leng(2016)]{Wang2016high}
X.~Wang and C.~Leng.
\newblock High dimensional ordinary least squares projection for screening
  variables.
\newblock \emph{J. R. Stat. Soc. Ser. B Stat. Methodol.}, pages 589--611, 2016.

\bibitem[Yu and Bien(2019)]{Yu2019estimating}
G.~Yu and J.~Bien.
\newblock Estimating the error variance in a high-dimensional linear model.
\newblock \emph{Biometrika}, 106\penalty0 (3):\penalty0 533--546, 2019.

\bibitem[Zhang(2010)]{MCP2010}
C.-H. Zhang.
\newblock Nearly unbiased variable selection under minimax concave penalty.
\newblock \emph{Ann. Statist.}, 38\penalty0 (2):\penalty0 894--942, 2010.

\bibitem[Zhang and Zhang(2014)]{Zhang2014}
C.-H. Zhang and S.~S. Zhang.
\newblock Confidence intervals for low dimensional parameters in high
  dimensional linear models.
\newblock \emph{J. R. Stat. Soc. Ser. B. Stat. Methodol.}, 76\penalty0
  (1):\penalty0 217--242, 2014.

\bibitem[Zhang et~al.(2019)Zhang, Jiang, and Lan]{Zhang19}
N.~Zhang, W.~Jiang, and Y.~Lan.
\newblock On the sure screening properties of iteratively sure independence
  screening algorithms.
\newblock \emph{arXiv:1812.01367}, 2019.

\bibitem[Zhang and Cheng(2017)]{Zhang2017}
X.~Zhang and G.~Cheng.
\newblock Simultaneous inference for high-dimensional linear models.
\newblock \emph{J. Amer. Statist. Assoc.}, 112\penalty0 (518):\penalty0
  757--768, 2017.

\bibitem[Zhang et~al.(2016)Zhang, Wu, Wang, and Li]{Zhang2016variable}
X.~Zhang, Y.~Wu, L.~Wang, and R.~Li.
\newblock Variable selection for support vector machines in moderately high
  dimensions.
\newblock \emph{J. R. Stat. Soc. Ser. B Stat. Methodol.}, 78\penalty0
  (1):\penalty0 53, 2016.

\bibitem[Zhao and Yu(2006)]{ZhaoYu2006}
P.~Zhao and B.~Yu.
\newblock On model selection consistency of {L}asso.
\newblock \emph{J. Mach. Learn. Res.}, 7:\penalty0 2541--2563, 2006.

\bibitem[Zhao et~al.(2019)Zhao, Yang, and He]{Zhao2019}
P.~Zhao, Y.~Yang, and Q.-C. He.
\newblock Implicit regularization via hadamard product over-parametrization in
  high-dimensional linear regression.
\newblock \emph{arXiv:1903.09367}, 2019.

\bibitem[Zhu et~al.(2011)Zhu, Li, Li, and Zhu]{Zhu2011model}
L.-P. Zhu, L.~Li, R.~Li, and L.-X. Zhu.
\newblock Model-free feature screening for ultrahigh-dimensional data.
\newblock \emph{J. Amer. Statist. Assoc.}, 106\penalty0 (496):\penalty0
  1464--1475, 2011.

\bibitem[Zou(2006)]{ALASSO2006}
H.~Zou.
\newblock The adaptive lasso and its oracle properties.
\newblock \emph{J. Amer. Statist. Assoc.}, 101\penalty0 (476):\penalty0
  1418--1429, 2006.

\end{thebibliography}

\end{singlespace}
\newpage 

\end{document}